\documentclass[final,onefignum,onetabnum,reqno]{amsart}

\usepackage{hyperref}
\hypersetup{pdftex,colorlinks=true,allcolors=blue}

 \usepackage{amsaddr}

\usepackage{amssymb,amsfonts,latexsym,amsmath,wrapfig,graphicx,amsthm}

\usepackage{bm}

\usepackage{enumitem}

\usepackage{chngcntr}
\usepackage{xcolor}
\usepackage{braket}

\usepackage{bbold}

\usepackage{tikz}
\usepackage{tikz-3dplot}
\usepackage{tikz-cd} 
\usetikzlibrary{patterns}

\usepackage[makeroom]{cancel}

\usepackage[yyyymmdd,hhmmss]{datetime}
\usepackage[all]{background}
\usepackage{soul}

\usepackage{mathrsfs}

\usepackage[capitalise]{cleveref}
\crefname{equation}{}{}

\SetBgContents{}
\undef\ul

\newtheorem{theorem}{Theorem}[section]

\newtheorem{lemma}[theorem]{Lemma}
\newtheorem{corollary}[theorem]{Corollary}
\newtheorem{proposition}[theorem]{Proposition}
\theoremstyle{remark}
\newtheorem{remark}{Remark}

\DeclareMathOperator{\Tr}{Tr}

\DeclareMathOperator{\Real}{Re}
\DeclareMathOperator{\Imag}{Im}

\DeclareMathOperator{\rot}{rot}

\undef\div
\DeclareMathOperator{\div}{div}

\DeclareMathOperator{\supp}{supp}

\definecolor{myblue}{RGB}{0, 100, 250}

\newcommand{\RR}{\mathbb{R}}
\newcommand{\CC}{\mathbb{C}}
\newcommand{\NN}{\mathbb{N}}

\newcommand{\ZZ}{\mathbb{Z}}

\newcommand{\TF}{\mathrm{TF}}

\newcommand{\ol}{\overline}
\newcommand{\ul}{\underline}

\renewcommand{\epsilon}{\varepsilon}
\renewcommand{\phi}{\varphi}

\renewcommand{\setminus}{\smallsetminus}

\newcommand{\wt}[1]{\widetilde{#1}}

\newcommand{\grad}{\bm{\nabla}}
\newcommand{\lapl}{\bm{\Delta}}

\newcommand{\vertiii}[1]{{\left\vert\kern-0.25ex\left\vert\kern-0.25ex\left\vert #1
   \right\vert\kern-0.25ex\right\vert\kern-0.25ex\right\vert}}

\newcommand{\UEG}{\mathrm{UEG}}
\newcommand{\LT}{\mathrm{LT}}
\newcommand{\LO}{\mathrm{LO}}

\newcommand{\xc}{\mathrm{xc}}

\newcommand{\TC}{\mathcal{T}}

\newcommand{\dd}{\mathrm{d}}

\newcommand{\iden}{\mathbb{1}}

\newcommand{\va}{\mathbf{a}}

\newcommand{\ve}{\mathbf{e}}

\newcommand{\vk}{\mathbf{k}}

\newcommand{\vx}{\mathbf{x}}
\newcommand{\vy}{\mathbf{y}}
\newcommand{\vz}{\mathbf{z}}
\newcommand{\vv}{\mathbf{v}}

\newcommand{\vu}{\mathbf{u}}
\newcommand{\vw}{\mathbf{w}}

\newcommand{\vX}{\mathbf{X}}
\newcommand{\vY}{\mathbf{Y}}

\newcommand{\vA}{\mathbf{A}}
\newcommand{\vB}{\mathbf{B}}

\newcommand{\vD}{\mathbf{D}}

\newcommand{\vR}{\mathbf{R}}

\newcommand{\vT}{\mathbf{T}}
\newcommand{\vM}{\mathbf{M}}

\newcommand{\vjp}{\mathbf{j}^{\mathrm{p}}}

\newcommand{\vnu}{{\bm{\nu}}}
\newcommand{\vzeta}{{\bm{\zeta}}}

\newcommand{\vtau}{{\bm{\tau}}}
\newcommand{\vomega}{{\bm{\omega}}}
\newcommand{\vGamma}{{\bm{\Gamma}}}

\newcommand{\vpi}{{\bm{\pi}}}

\newcommand{\dvx}{\dd\mathbf{x}}

\renewcommand{\le}{\leqslant}
\renewcommand{\ge}{\geqslant}

\newcommand{\Hil}{\mathfrak{H}}
\newcommand{\Fock}{\mathfrak{F}}

\newcommand{\OC}{\mathcal{O}}
\newcommand{\Cc}{\mathcal{C}}

\newcommand{\SF}{\mathfrak{S}}

\newcommand{\rdens}{\mathcal{I}}
\newcommand{\rcdens}{\mathcal{J}}

\newcommand{\Lloc}{L_{\mathrm{loc}}}

\newcommand{\pure}{\mathrm{pure}}
\newcommand{\mixed}{\mathrm{mixed}}

\newcommand{\tetra}{\mathbb{\Delta}}

\newcommand{\dm}{\mathcal{D}}

\newcommand{\SO}{\mathrm{SO}}
\newcommand{\ketbra}[2]{\ket{#1}\!\!\bra{#2}}
\def\Xint#1{\mathchoice
{\XXint\displaystyle\textstyle{#1}}%
{\XXint\textstyle\scriptstyle{#1}}%
{\XXint\scriptstyle\scriptscriptstyle{#1}}%
{\XXint\scriptscriptstyle\scriptscriptstyle{#1}}%
\!\int}
\def\XXint#1#2#3{{\setbox0=\hbox{$#1{#2#3}{\int}$ }
\vcenter{\hbox{$#2#3$ }}\kern-.6\wd0}}

\def\dashint{\Xint-}

\title[The Magnetic Uniform electron gas]{Thermodynamic limit for the magnetic uniform electron gas and representability of density-current pairs}

\author{Mih\'aly A. Csirik \and Andre Laestadius}
\address{Department of Computer Science, Oslo Metropolitan University, Norway\\ Hylleraas Centre for Quantum Molecular Sciences, Department of Chemistry, University of Oslo, Norway}
\author{Erik I. Tellgren}
\address{Hylleraas Centre for Quantum Molecular Sciences, Department of Chemistry, University of Oslo, Norway}
  
\begin{document}

\maketitle

\begin{abstract}
Although the concept of the uniform electron gas is essential to quantum physics, it has only been defined recently in a rigorous manner by Lewin, Lieb and Seiringer. We extend their approach to include the magnetic case, by which we
mean that the vorticity of the gas is also held constant. Our definition involves the grand-canonical version of the universal functional
introduced by Vignale and Rasolt in the context of current-density-functional theory. Besides establishing the existence of the thermodynamic limit, we derive an estimate on the kinetic energy functional
that also gives a convenient answer to the (mixed) current-density representability problem.
\end{abstract}




\section{Introduction}

Following \cite{lewin2018statistical}, we view the uniform electron gas (UEG) as a continuous system of (interacting or not) electrons with prescribed constant density everywhere. 
In the non-interacting case,
one can readily write down the ground state of the UEG via a fully occupied Fermi sphere. One can consider the classical UEG by neglecting the kinetic energy of the
electrons \cite{lewin2018statistical}. 

In contrast, the Jellium is a continuous system of interacting electrons in a constant neutralizing potential \cite{lieb1975thermodynamic}. In the classical case, it was
recently shown that the Jellium and the UEG energies are the same \cite{lewin2019floating}. In the quantum case, the problem is still open. We will solely focus on the UEG in this article.

The article \cite{lewin2018statistical} pioneered the use of universal density functionals of density-functional theory (DFT) as a theoretical tool in quantum statistical mechanics, namely
for the \emph{definition} of the UEG. In a later article \cite{lewin2020local}, the same authors established an even more elegant (and equivalent)
definition of the UEG, and we will restrict ourselves to this latter approach. 

When magnetic fields are present, the situation is much more delicate \cite{giuliani2005quantum}. Even the noninteracting UEG is more difficult to analyze 
in the magnetic case \cite{fushiki1992matter,lieb1994asymptoticsI,lieb1994asymptoticsII}, even though the ground state can be given exactly.
The celebrated Laughlin wave function has also received attention
from the mathematics community \cite{rougerie2014quantum,rougerie2019laughlin}, which is a rather precise trial state (with almost constant density in the thermodynamic limit) 
for the Jellium in a constant magnetic field.

\vspace{1em}
\noindent\emph{Acknowledgements.} A.L. and M.A.Cs. acknowledge funding through ERC StG REGAL under agreement No.\ 101041487. 
M.A.Cs. is grateful to Mathieu Lewin and Giovanni Vignale for helpful discussions. 
A.L., M.A.Cs. and E.I.T. acknowledge funding from RCN through the Hylleraas Centre 
 under agreement No.\ 262695. E.I.T. acknowledges funding from RCN under agreement No.\ ~287950.

\section{Setting}

\subsection{States}

In the first half of this section, we will fix the particle number $N\in\NN$.
The \emph{$N$-particle fermionic Hilbert space} is the antisymmetric power
$\Hil^N=\bigwedge^N \Hil$ (endowed with the usual inner product), where $\Hil=L^2(\RR^d)$
is the \emph{one-particle Hilbert space}. Note that $\Hil^1=\Hil$. 

Generally speaking, a density matrix (or state) on a Hilbert space is self-adjoint, non-negative trace class operator that has trace one. We first consider the Hilbert space to be the fermionic $N$-particle space $\Hil^N$.
To make notations more compact, we henceforth restrict our attention to states with finite kinetic energy, which we call \emph{$H^1$-regular states}.
The \emph{set of $H^1$-regular (fermionic) $N$-particle states} is defined as
$$
\dm_N=\left\{ \begin{aligned}
\Gamma&\in\SF_1(\Hil^N) : \Gamma=\sum_{j\ge 1} \mu_j\Gamma_{\Psi_j}, \;\text{for some}\;\mu_j\ge 0\;\text{with}\;\sum_{j\ge 1}\mu_j=1,\,\text{and}\\
&\{\Psi_j\}\subset\Hil^N \cap H^1(\RR^{dN})\,\text{is}\,L^2\text{-orthonormal with}\;\sum_{j\ge 1} \mu_j \|\grad\Psi_j\|^2<\infty
\end{aligned}\right\},
$$
where $\SF_1$ the trace class.
Then clearly $0\le\Gamma=\Gamma^\dag$ and $\Tr_{\Hil^N}\Gamma=1$ for all $\Gamma\in\dm_N$. It is not hard to show that the operator $\Gamma$ can be given by the Hermitian integral kernel
$\Gamma : \RR^{dN}\times\RR^{dN}\to\CC$ (denoted with the same symbol), 
$$
\Gamma(\vX;\vY)=\sum_{j\ge 1} \mu_j \Psi_j(\vX) \ol{\Psi_j(\vY)},
$$
so that its action on an element $\Phi\in\Hil^N$ is given by $\Gamma\Phi(\vX)=\int \Gamma(\vX;\vY)\Phi(\vY)\,\dd\vY$. 

Let $\mathcal{A}:\Hil^N\to\RR$ be a nonnegative quadratic form with 
domain in $\Hil^N\cap H^1(\RR^{dN})$, and let $A:\Hil^N\to\Hil^N$ be its self-adjoint realization. Then for any $\Gamma\in\dm_N$, we define
$$
\Tr_{\Hil^N} A\Gamma:=\sum_{j\ge 1} \mu_j \mathcal{A}(\Psi_j),
$$
which is a nonnegative number or $+\infty$. Choosing $A=-\lapl$, i.e. $\mathcal{A}(\Psi)=\|\grad\Psi\|^2$, we find that $\Tr_{\Hil^N}(-\lapl\Gamma)=\sum_{j\ge 1}\mu_j \|\grad\Psi_j\|^2<\infty$
for any $\Gamma\in\dm_N$. 

Next, we consider states on Fock space. For convenience, we introduce the notation $\Hil^0=\CC$. The \emph{fermionic Fock space} $\Fock$ is then the Hilbert space given by direct sum of the Hilbert spaces $\Hil^n$, $\Fock=\Hil^0\oplus\Hil^1\oplus\Hil^2\oplus\ldots$ endowed with the usual inner product.
A \emph{Fock space state} $\vGamma$ is a positive self-adjoint, trace-class operator on $\Fock$ with unit trace: $\vGamma\in\SF_1(\Fock)$, $0\le\vGamma=\vGamma^\dag$  and
$\Tr_\Fock\vGamma=1$. 
We will not need this generality, however, and it will suffice to consider states $\vGamma$ which commute with the number operator.
It is easy to show by induction that such states are \emph{block-diagonal:} 
$$
\vGamma=\Gamma_0\oplus\Gamma_1\oplus\ldots,
$$
where $\Gamma_n\in\SF_1(\Hil^n)$ is such that $0\le \Gamma_n=\Gamma_n^\dag$.
Note in particular that $\Gamma_0\in\RR_+$, and that the normalization constraint now reads $\sum_{n\ge 0} \Tr_{\Fock}\Gamma_n=1$.
Again, we will further restrict attention to states with finite kinetic energy, which we called $H^1$-regular states. All this leads to the definition
$$
\dm=\left\{ \begin{aligned}
\vGamma=\Gamma_0\oplus\Gamma_1\oplus\ldots\in\SF_1(\Fock) :&\, \Gamma_n\in\SF_1(\Hil^n) \,\text{and}\, 0\le \Gamma_n=\Gamma_n^\dag \,\text{for all}\,n\ge 0, \\
&\sum_{n\ge 0} \Tr\Gamma_n=1\;\text{and}\; \sum_{n\ge 0}\Tr_{\Hil^n} (-\lapl\Gamma_n)<\infty
\end{aligned}\right\}
$$
of the set of \emph{$H^1$-regular block-diagonal Fock space states}. Each $\Gamma_n$ may be given by an integral kernel $\Gamma_n(\vX_n;\vY_n)$,
like before (of course these will not be normalized to have unit trace, in general).
Note also that by linearity, the one-, and the two-particle reduced density matrices of $\vGamma=\Gamma_0\oplus\Gamma_1\oplus\ldots\in\dm$ are given by
$$
\gamma_\vGamma=\sum_{n\ge 0} \gamma_{\Gamma_n}, \quad\text{and}\quad \gamma^{(2)}_\vGamma=\sum_{n\ge 0} \gamma^{(2)}_{\Gamma_n}.
$$
We note in passing that all derived quantities, such as $\rho_\vGamma$, may be defined by linearity.

\subsection{One-particle quantities}
In this section, we introduce various quantities which may be defined using the one-particle density matrix. 
The \emph{one-particle reduced density matrix $\gamma\in\SF_1(\Hil)$ of $\Gamma\in\dm_N$} is defined as ($N$ times) the partial trace of $\Gamma$ via duality:
$$
\Tr_\Hil (\gamma B)=N\Tr_{\Hil^N} (\Gamma (B\otimes \iden_{\Hil^{N-1}})) \quad\text{for all}\quad B\in\mathcal{B}(\Hil).
$$
In symbols, this operation is sometimes written as $\gamma=N\Tr_{\Hil^{N-1}} \Gamma$. Note the choice of normalization, so that we have $\Tr_\Hil\gamma=N$.
It is easy to see that $\gamma$ can be given by the kernel
$$
\gamma(\vx,\vy)=N\int_{\RR^{d(N-1)}} \Gamma(\vx,\vx_2,\ldots,\vx_N;\vy,\vx_2,\ldots,\vx_N)\,\dvx_2\cdots\dvx_N.
$$
It is also clear that $\gamma$ is a trace class operator on the one-particle Hilbert space $\Hil$ such that $0\le\gamma=\gamma^\dag\le\iden$.
 By the spectral theorem $\gamma$ has the decomposition
$\gamma=\sum_{j\ge 1} \lambda_j \ketbra{\phi_j}{\phi_j}$,
which in terms of its kernel can be written as
$$
\gamma(\vx,\vy)=\sum_{j\ge 1} \lambda_j \phi_j(\vx)\ol{\phi_j(\vy)}
$$
for some $0\le \lambda_j\le 1$ with $\sum_{j\ge 1}\lambda_j=N$ and $\{\phi_j\}_{j\ge 1}\subset H^1(\RR^d)$ $L^2$-orthonormal.
It is clear that $\gamma$ has finite kinetic energy, i.e. it is $H^1$-regular: $\sum_{j\ge 1} \lambda_j \|\grad\phi_j\|^2<\infty$.

The \emph{density} $\rho_\gamma\in L^1(\RR^d)$ of the one-particle density matrix $\gamma$ may also be defined via duality, 
i.e. by requiring that $\int_{\RR^d} \rho_\gamma \phi=\Tr_\Hil(\gamma\phi)$ for all $\phi\in L^\infty(\RR^d)$. This gives
$$
\rho_\gamma(\vx)=\sum_{j\ge 1} \lambda_j |\phi_j(\vx)|^2.
$$ 
The kinetic energy $\Tr_\Hil(-\lapl\gamma):=\sum_{j\ge 1} \lambda_j \|\grad\phi_j\|^2<\infty$ may be expressed using the kernel $\gamma(\vx,\vy)$ as
$$
\Tr(-\lapl\gamma)=\int_{\RR^d} \grad_\vx\cdot\grad_{\vy} \gamma(\vx,\vy)|_{\vy=\vx}\,\dvx.
$$
Here, the integrand $\tau_\gamma(\vx):=\grad_\vx\cdot\grad_{\vy} \gamma(\vx,\vy)|_{\vy=\vx}$ is called the \emph{kinetic energy density}, and more generally
$$
\vtau_\gamma(\vx):=\grad_\vx\otimes \grad_\vy \gamma(\vx,\vy)|_{\vy=\vx}=\sum_{j\ge 1} \lambda_j \grad\phi_j(\vx)\otimes \ol{\grad\phi_j(\vx)}
$$
the \emph{kinetic energy tensor}. Clearly, $0\le \vtau_\gamma(\vx)=\vtau_\gamma(\vx)^\dag$ and $\Tr_{\RR^3} \vtau_\gamma(\vx)=\tau_\gamma(\vx)$.

A further useful quantity is the \emph{complex current density} $\vzeta_\gamma:\RR^d\to\CC^d$, which is given by
\begin{equation}\label{vzetadef}
\vzeta_\gamma(\vx):= \grad_\vx \gamma(\vx,\vy)|_{\vy=\vx}=\sum_{j\ge 1} \lambda_j\grad\phi_j(\vx)\ol{\phi_j(\vx)}.
\end{equation}
A straightforward calculation shows that 
\begin{equation}\label{vzetareim}
\Real\vzeta_\gamma(\vx)=\frac{1}{2} \grad\rho_\gamma(\vx),\quad\text{and}\quad \vjp_\gamma(\vx):=\Imag\vzeta_\gamma(\vx)=\Imag\grad_\vx \gamma(\vx,\vy)|_{\vy=\vx},
\end{equation}
where $\vjp_\gamma:\RR^d\to\RR^d$ is called the \emph{paramagnetic current density}. 

The following theorem is based on \cite{DOBSON_JPCM4_7877,sen2018local,laestadius2019kohn}.
\begin{proposition}\label{basicprop}
Suppose that $\gamma$ is fermionic one-particle density matrix composed of sufficiently smooth orbitals. Then the following hold true.
\begin{itemize}
\item[(i)] $|\grad\sqrt{\rho_\gamma}(\vx)|^2 + \frac{|\vjp_\gamma(\vx)|^2}{\rho(\vx)} + \frac{1}{\sqrt{d}}\rho(\vx)\left|\vD_a\frac{\vjp_\gamma(\vx)}{\rho_\gamma(\vx)}\right| \le \tau_\gamma(\vx)$ for a.a. $\vx\in\RR^d$.
\item[(ii)] $|\vjp_\gamma(\vx)|\le|\vzeta_\gamma(\vx)|\le \tau_\gamma(\vx)^{1/2} \rho_\gamma(\vx)^{1/2}$ for a.a. $\vx\in\RR^d$, hence $\vjp_\gamma\in L^1(\RR^d)^d$.
\end{itemize}
\end{proposition}
Here, the symmetric-, and antisymmetric derivatives $\vD_s$ and $\vD_a$ will be used at some occasions, given by
$$
\vD_s\vu=\frac{\vD\vu + (\vD\vu)^\top}{2}\quad\text{and}\quad \vD_a\vu=\frac{\vD\vu - (\vD\vu)^\top}{2},
$$
for any differentiable vector field $\vu:\RR^d\to\RR^d$ and $\vD\vu:\RR^d\to\RR^{d\times d}$ is its derivative. Further, we denote by $\vA:\vB=\Tr(\vA^\top \vB)$ the Frobenius inner product, and by $|\vA|=\sqrt{\vA:\vA}$
the Frobenius norm. 

Let us consider the case $d=3$. Then $\left|\vD_a\vu\right|=\frac{\sqrt{2}}{2}|\rot\vu|$ and part (i) improves the usual Hoffmann--Ostenhof bound \cite{hoffmann1977schrodinger} as
\begin{equation}\label{kinweizgauge}
\Tr(-\lapl\gamma)\ge\int_{\RR^3} |\grad\sqrt{\rho_\gamma}|^2 + \int_{\RR^3} \frac{|\vjp_\gamma|^2}{\rho_\gamma} + \frac{1}{\sqrt{6}}\int_{\RR^3} \rho_\gamma\left|\vnu_\gamma\right|,
\end{equation}
for any $H^1$-regular density matrix $\gamma$, where $\vnu_\gamma = \rot\frac{\vjp_\gamma}{\rho_\gamma}$ is the (gauge-invariant, see below) \emph{vorticity}. 
We call the terms on the r.h.s. the \emph{von Weizs\"acker-}, the \emph{gauge}-, and the \emph{vorticity term}, respectively.
Note that \cref{kinweizgauge} implies that for $H^1$-regular density one-particle matrices $\gamma$, we have $\int_{\RR^3} |\grad\sqrt{\rho_\gamma}|^2<\infty$, $\int_{\RR^3} \frac{|\vjp_\gamma|^2}{\rho_\gamma}<\infty$
and $\int_{\RR^3} \rho_\gamma\left|\vnu_\gamma\right|<\infty$, quantities which play important roles later. Finally, we note that for $d=3$ the constant in front of vorticity term in \cref{kinweizgauge} can be improved to 1~\cite{sen2018local}.

\begin{remark}
Part (ii) of \cref{basicprop} implies that the support of $\vjp$ is contained in the support of $\rho$, in particular, the quotients $\frac{\vjp}{\sqrt{\rho}}$ and 
$\frac{|\vjp|^2}{\rho}$ are meaningful on $\supp\rho$. Moreover, if $\tau\equiv 0$ in some ball, then by (i), $\rho$ is constant and
$\vjp\equiv 0$ on that ball. 
\end{remark}

\subsection{Gauge transformation}\label{gaugesec}

We now discuss a very important tool in the study of magnetic systems.
Given an $N$-particle state $\Gamma\in\dm_N$, consider the \emph{gauge-transformed} state
$$
\wt{\Gamma}(\vX,\vY)=e^{i\sum_{1\le j\le N} ( g(\vy_j)- g(\vx_j))}\Gamma(\vX,\vY)
$$
where $g$ is called the \emph{gauge function}.

\begin{proposition}\label{densgauge}
Let $\gamma$ be a fermionic one-particle density matrix. Then the gauge transformation of $\gamma$ is given by
$$
\wt{\gamma}(\vx,\vy)=e^{i( g(\vy)- g(\vx))} \gamma(\vx,\vy).
$$
In particular, $\rho_{\wt{\gamma}}=\rho_\gamma$.
Furthermore, $\vjp_{\wt{\gamma}}=\vjp_\gamma - \rho_\gamma \grad g$ and 
\begin{equation}\label{kingauge}
\Tr(-\lapl\wt{\gamma})=\Tr(-\lapl\gamma) - 2\int_{\RR^d} \vjp_\gamma \cdot \grad g + \int_{\RR^d} \rho_\gamma |\grad g|^2.
\end{equation}
\end{proposition}

We did not specify what function space the gauge function $ g$ is coming from. In order for the 
r.h.s. of \cref{kingauge} to make sense, it is enough to assume that $\grad g\in L_{\rho_\gamma}^2(\RR^d)^d$, using the Cauchy--Bunyakovsky--Schwarz inequality and 
 \cref{kinweizgauge}. 

Note that the gauge term $\int_{\RR^d} \frac{|\vjp|^2}{\rho}$ transforms exactly like the kinetic energy \cref{kingauge}. More precisely,
$$
\int_{\RR^d} \frac{|\vjp_{\wt{\gamma}}|^2}{\rho_{\wt{\gamma}}}=\int_{\RR^d} \frac{|\vjp_{\gamma}|^2}{\rho_\gamma} - 2\int_{\RR^d} \vjp_\gamma\cdot \grad g +\int_{\RR^d} \rho_\gamma |\grad g|^2,
$$
therefore the inequality  \cref{kinweizgauge} is ``consistent'' with a gauge transformation in the sense that all the terms involving $ g$ will cancel.
We call such inequalities \emph{gauge-invariant}.

\subsection{Representability}\label{repsec}
Fix $N\in\NN$. Recall that the set of \emph{pure (or mixed) state $N$-representable densities} may be written as
$$
\rdens_N=\Bigg\{ \rho\in L^1(\RR^d;\RR_+) : \grad\sqrt{\rho}\in L^2(\RR^d)^d,\; \int_{\RR^d}\rho=N \Bigg\}.
$$
according to \cite{harriman1981orthonormal} and \cite[Theorem 1.2]{lieb1983density}. For later purposes, we also introduce the set of \emph{Fock space-representable densities} as
\begin{align*}
\rdens&=\left\{ \rho\in L^1(\RR^d;\RR_+) : \grad\sqrt{\rho}\in L^2(\RR^d)^d \right\}.
\end{align*}

For density-paramagnetic current density pairs, the situation is more complicated. Consider the set 
$$
\rcdens_N^\pure=\Bigl\{(\rho,\vjp) : \rho=\rho_\Psi,\, \vjp=\vjp_\Psi,\; \text{for some} \; \Psi\in\Hil^N\cap H^1(\RR^{dN}),\; \|\Psi\|=1 \Bigr\}
$$
of all \emph{pure state $N$-representable density-current pairs}, and
$$
\rcdens_N^\mixed=\{(\rho,\vjp) : \rho=\rho_\Gamma,\,\vjp=\vjp_\Gamma,\; \text{for some}\;\Gamma\in\dm_N \},
$$
the set of \emph{mixed state $N$-representable density-current pairs}.
In this more general setting it is yet unknown whether $\rcdens_N^\pure$ is equal to $\rcdens_N^\mixed$ or not.
Unfortunately, no explicit characterization of either $\rcdens_N^\pure$ or $\rcdens_N^\mixed$ exists up to date. We can specify a larger set by setting
\begin{equation}\label{rdensNdef}
\rcdens_N'=\Bigg\{ (\rho,\vjp) : \rho\in\rdens_N, \; \vjp\in L^1(\RR^d)^d,\;  \int_{\RR^d} \frac{|\vjp|^2}{\rho} <\infty \Bigg\},
\end{equation}
so that by \cref{kinweizgauge}, $\rcdens_N^\pure,\rcdens_N^\mixed\subset\rcdens_N'$.
It can be shown that $\rcdens_N'$ is also convex \cite{laestadius2014density}. We remark that the finiteness of the vorticity term in \cref{kinweizgauge}
may also be added to the definition of $\rcdens_N'$.

Whether any $(\rho,\vjp)\in\rcdens_N'$ can be represented with a Slater determinant $\Phi\in\Hil^N\cap H^1(\RR^{3N})$ such that $\rho_\Phi=\rho$ and $\vjp_\Phi=\vjp$
(the \emph{determinantal density-current $N$-representability problem})
is still not settled completely. The following result gives a partial answer, however.
\begin{theorem}(Lieb--Schrader \cite{lieb2013current})\label{liebschrader}
Let $d=3$ and suppose that $(\rho,\vjp)\in\rcdens_N'$ and set $\vnu=\rot \frac{\vjp}{\rho}$. Then the following holds true.
\begin{itemize}
\item[(i)] (Zero vorticity) If $\vnu=0$, then there exists a Slater determinant $\Phi\in\Hil^N\cap H^1(\RR^{3N})$, such that $\rho_\Phi=\rho$, $\vjp_\Phi=\vjp$ and
$$
\|\grad\Phi\|^2\le C_N\int_{\RR^3} |\grad \sqrt{\rho}|^2 + \int_{\RR^3} \frac{|\vjp|^2}{\rho}
$$
for some constant $C_N>0$.
\item[(ii)] If $N\ge 4$, then there exists a Slater determinant $\Phi\in\Hil^N$, such that $\rho_\Phi=\rho$, $\vjp_\Phi=\vjp$.
Moreover, if there is a $\delta>0$ such that the growth conditions
\begin{align*}
\sup_{\vx\in\RR^3} f(\vx) | \vnu(\vx) | < \infty, \quad \text{and} \quad \sup_{\vx\in\RR^3} f(\vx) | \vD \vnu(\vx) | < \infty
\end{align*}
hold true, then $\|\grad\Phi\|<\infty$. Here, we have set $f(\vx)=(1+x_1^2)^{\frac{1+\delta}{2}} (1+x_2^2)^{\frac{1+\delta}{2}} (1+x_3^2)^{\frac{1+\delta}{2}}$.
\end{itemize}
\end{theorem}
Moreover, \cite{lieb2013current} also exhibits an example $(\rho,\vjp)$ in the case $N=2$ (and of course $\vnu\neq 0$), when there is no $C^1$-solution
to the representability problem. Unfortunately, the restriction $N\ge 4$ and the fact that no control over the kinetic energy is retained in (ii) renders the preceding result inapplicable for our case.

The following theorem establishes mixed representability under different assumptions, now for general $N\in\NN$.
\begin{theorem}(Tellgren--Kvaal--Helgaker \cite{tellgren2014fermion})\label{erikrep}
Let $d=3$ and suppose that $(\rho,\vjp)\in\rcdens_N'$ and that
\begin{equation}\label{erikcond}
\int_{\RR^3} (1+|\vx|^2)\rho \Bigl| \vD \frac{\vjp}{\rho} \Bigr |^2\,\dvx<\infty.
\end{equation}
Then there exists a constant $C>0$ and a fermionic one-particle density matrix $\gamma$ such that $\rho_\gamma=\rho$, $\vjp_\gamma=\vjp$ and
\begin{equation}\label{erikineq}
\Tr (-\lapl\gamma)\le \int_{\RR^3} |\grad \sqrt{\rho}|^2 + CN + \int_{\RR^3} \frac{|\vjp|^2}{\rho} + C \int_{\RR^3} (1+|\vx|^2)\rho \Bigl|\vD \frac{\vjp}{\rho} \Bigr|^2\,\dvx.
\end{equation}
\end{theorem}

Representability holds also true without the integer particle number constraint on the density, i.e. when $\int_{\RR^d}\rho\not\in\NN$, but the representing state will live in
 the Fock space $\Fock$.
Analogously to the $N$-particle case, we define the set of \emph{Fock space-representable density-current pairs} as
$$
\rcdens:=\Bigg\{(\rho,\vjp) : \rho=\rho_\vGamma,\, \vjp=\vjp_\vGamma\; \text{for some} \; \vGamma\in\dm\Bigg\}.
$$
Then the conditions of \cref{erikrep} provide a smaller set, but we will use an improved version below. 

\subsection{Definition of the grand-canonical current-density functionals}

In this section, we define the current-density functionals relevant to our study. These are well-known in CDFT,
but we will also employ them in the definition of the uniform electron gas. In this section we restrict ourselves to the physically most
relevant 3D case.

We begin by introducing the \emph{grand-canonical Vignale--Rasolt functional}~\cite{Vignale1988} via
\begin{equation}\label{vrdef}
\boxed{F(\rho,\vjp):=\inf_{\substack{\vGamma\in\dm\\ \rho_{\vGamma}=\rho\\  \vjp_{\vGamma}=\vjp}} \sum_{n\ge 1} \Tr_{\Hil^n}\Biggl( -\sum_{1\le j\le n} \lapl_{\vx_j} + \sum_{1\le j<k\le n} \frac{1}{|\vx_j-\vx_k|} \Biggr)\Gamma_n}.
\end{equation}
Omitting the Coulomb interaction in the above definition, we obtain the \emph{mixed-state kinetic energy functional} $T(\rho,\vjp)$.
By definition, $T$ and $F$ is (jointly) convex on $\rcdens$.
In our analysis we will consider the \emph{indirect energy}
\begin{equation}\label{inddef}
\boxed{E(\rho,\vjp):=F(\rho,\vjp)-D(\rho)}
\end{equation}
where for any finite measures $\mu_1$ and $\mu_2$, we define the \emph{Coulomb inner product} via
$$
D(\mu_1,\mu_2):=\frac{1}{2} \iint_{\RR^3\times\RR^3} \frac{\mu_1(\dd\vx)\mu_2(\dd\vy)}{|\vx-\vy|},
$$
and additionally, as an abuse of notation, we set $D(\mu):=D(\mu,\mu)$ for the \emph{direct energy of $\mu$}. If the measures $\mu_j$ come from densities $\rho_j$,
i.e. $\dd\mu_j=\rho_j\,\dd\vx$ ($j=1,2$), then we will use the notations $D(\rho_1,\rho_2)$ and $D(\rho)$.

For convenience, we introduce the \emph{gauge-corrected indirect energy}
$$
\boxed{\ol{E}(\rho,\vjp):=E(\rho,\vjp)-\int_{\RR^3}\frac{|\vjp|^2}{\rho}}
$$
and similarly $\ol{T}(\rho,\vjp)$. We will also consider the \emph{exchange-correlation functional}
$$
\boxed{E^{\xc}(\rho,\vjp):=E(\rho,\vjp)-T(\rho,\vjp)}
$$
which is more common in the physics and chemistry literature.

The reason for the introduction of the ``gauge-corrected'' quantities is the following. Note that under a gauge transformation $\wt{\vjp}=\vjp-\rho\grad g$, the functional $F$ transforms as
\begin{equation}\label{vrgauge}
F(\rho,\wt{\vjp})=F(\rho,\vjp) - 2\int_{\RR^3} \vjp\cdot\grad g + \int_{\RR^3} \rho|\grad g|^2,
\end{equation}
i.e. like the kinetic energy (see \cref{densgauge}). Analogous relations holds for $T(\rho,\vjp)$ and $E(\rho,\vjp)$.
It follows that the gauge-corrected indirect-, and kinetic energy $\ol{E}(\rho,\vjp)$ and $\ol{T}(\rho,\vjp)$, and exchange-correlation functional $E^{\xc}(\rho,\vjp)$ are all \emph{gauge-invariant}. 

By replacing the $N$-particle Hilbert space with the corresponding Fock space, the next result follows along similar lines as \cite[Theorem 1]{kvaal2021lower}.
\begin{theorem}\label{vropt}
Let $(\rho,\vjp)\in\rcdens$. Then the infima in the definitions \cref{vrdef} of the Vignale--Rasolt functional $F$ and the kinetic energy functional $T$ are both attained.
\end{theorem}

The grand-canonical Levy--Lieb functional $F(\rho)$ (no paramagnetic current density argument) and the kinetic energy functional $T(\rho)$ is defined similarly to $F(\rho,\vjp)$ and $T(\rho,\vjp)$, only without any constraint on $\vjp_\vGamma$ in the infimum.
\begin{proposition}[Correspondence with the nonmagnetic case]\label{corrthm}
For any $\rho\in\rdens$, the relations $T(\rho,0)=T(\rho)$ and $F(\rho,0)=F(\rho)$ hold true. In particular, 
$$
E(\rho,0)=E(\rho)\quad\text{and}\quad E^{\xc}(\rho,0)=E^{\xc}(\rho),
$$
where $E(\rho)=F(\rho)-D(\rho)$ and $E^{\xc}(\rho)=E(\rho)-T(\rho)$.
\end{proposition}
\begin{proof}
By convexity and time-reversal symmetry, we find $F(\rho,0)\le \tfrac{1}{2}F(\rho,\vjp)+\tfrac{1}{2}F(\rho,-\vjp)=F(\rho,\vjp)$. Hence, $F(\rho)\le F(\rho,0)\le \inf_{\vjp} F(\rho,\vjp)=F(\rho)$ as desired. The proof of the relation between the kinetic energy functionals is similar.
\end{proof}

\section{Main results}

\subsection{The magnetic uniform electron gas}\label{uegsec}

This section is devoted to the precise definition of the various energetic quantities of the 3D uniform electron gas (UEG) with constant density and constant vorticity.

Recall that the vorticity of the ground state of the \emph{noninteracting} UEG in a constant magnetic field $\vB_0$ is $\vnu_0=-\vB_0$.
Physically, we can generate a constant vorticity $\vnu_0$ in the noninteracting UEG by subjecting it to an external magnetic field $-\vB_0$.
Moreover, linear response theory suggests the same holds true to the first order in the \emph{interacting} case \cite{giuliani2005quantum}.
Instead of referring to the external magnetic field, we will simply demand the UEG to have constant density \emph{and} constant vorticity. 

A function $0\le \eta\in C_c^\infty(\RR^3)$ such that $\supp\eta\subset B_1$, $\int_{\RR^3} \vx\eta(\vx)\,\dd\vx=0$ and $\int_{\RR^3}\eta=1$ will be
called a \emph{regularization function}. Note that this is a slightly more general definition than the one in \cite{lewin2018statistical,lewin2020local},
where it is required that a regularization function be radial. We also introduce the notation $\eta_\delta(\vx)=\delta^{-3}\eta(\delta^{-1}\vx)$.

Define the \emph{indirect gauge-corrected energy per volume} of the (interacting) UEG at constant density $\rho_0\in\RR_+$ and constant vorticity $\vnu_0\in\RR^3$ inside a bounded domain $\Omega\subset\RR^3$ smeared with $\eta$ as
\begin{equation}\label{etetradef}
\boxed{e_{\Omega,\eta}(\rho_0,\vnu_0):=\frac{1}{|\Omega|}\ol{E}\big( (\iden_{\Omega}*\eta)(\rho_0,\tfrac{1}{2}\rho_0\vnu_0 \times\cdot)\big)}
\end{equation}
Note that $\rot \left(\tfrac{1}{2}\vnu_0\times\vx\right) = \vnu_0$ and that $e_{\Omega,\eta}(0,\vnu_0)=e_{\Omega,\eta}(0,0)=0$. The smearing of the indicator $\iden_\Omega$ is necessary, since $\int_{\RR^3} |\grad\sqrt{\iden_\Omega}|^2=\infty$ and this would render the energy $E$ infinite (see \cref{kinweizgauge}).
Recall that $\ol{E}$ contains a term of the form $-\int_{\RR^3} \frac{|\vjp|^2}{\rho}$, which we call the \emph{gauge correction term},
and it needs to be subtracted otherwise the energy per volume would diverge
in the thermodynamic limit. In fact, the gauge correction term grows faster than the volume,
\begin{align*}
\frac{1}{|\ell\Omega|}\int_{\RR^3} (\iden_{\ell\Omega} * \eta)(\vx)  |\vnu_0\times\vx|^2\,\dd\vx=C_{\Omega,\eta,\vnu_0} \ell^2,
\end{align*}
where $C_{\Omega,\eta,\vnu_0}>0$ is a constant depending only on $\Omega$, $\eta$ and $\vnu_0$.
Hence, by \cref{kinweizgauge}, the energy per volume diverges as $\ell\to\infty$. Notice that replacing 
$\tfrac{1}{2}\rho_0\vnu_0 \times\cdot\to \tfrac{1}{2}\rho_0\vnu_0 \times\cdot + \rho_0\grad g$
in the definition \cref{etetradef} above would yield the same value for $\ol{E}$ due to the gauge transformation rule \cref{vrgauge}.
In particular, one can choose the \emph{Landau gauge} instead of the \emph{symmetric gauge} $\tfrac{1}{2}\vnu_0\times\vx$ above without changing anything. 

By replacing $\ol{E}$ with $\ol{T}$ on the r.h.s. of \cref{etetradef}, we can also define the \emph{kinetic energy per volume} $\tau_{\Omega,\eta}(\rho_0,\vnu_0)$. 
Finally, we define the \emph{exchange-correlation energy per volume}
\begin{align*}
e_{\Omega,\eta}^{\xc}(\rho_0,\vnu_0):= e_{\Omega,\eta}(\rho_0,\vnu_0) - \tau_{\Omega,\eta}(\rho_0,\vnu_0).
\end{align*}
Note that, in this case the gauge correction terms cancel. 

Following Fischer, we say that a sequence of domains $\{\Omega_N\}_{N\in\NN}\subset\RR^3$ has a \emph{uniformly regular boundary} if
$$
|\partial\Omega_N + B_r|\le Cr|\Omega_N|^{2/3}\quad\text{for all}\quad r\le \frac{|\Omega_N|^{2/3}}{C},
$$
for some $C>0$. For instance, one may take $\Omega_N=N\Omega$ for some convex, bounded and open set \cite{HAINZL2009454}.

The following theorem is the main result of the paper. 
\begin{theorem}[Uniform electron gas]\label{thermothm}
Let $\rho_0\in\RR_+$, $\vnu_0\in\RR^3$
and let $\{\Omega_N\}_{N\in\NN}\subset\RR^3$ be a sequence of bounded domains such that $|\Omega_N|\to\infty$ and it has a uniformly regular boundary.
Let $\delta_N>0$ be any sequence such that $\delta_N/|\Omega_N|^{1/3}\to 0$ and $\delta_N|\Omega_N|^{1/3}\to\infty$.
Then the following thermodynamic limit exists
$$
e_\UEG(\rho_0,\nu_0):=\lim_{N\to\infty} e_{\Omega_N,\eta_{\delta_N}}(\rho_0,\vnu_0)
$$
and is independent of the choice of the sequences $\{\Omega_N\}$, $\{\delta_N\}$ and the regularization function $\eta$.
Analogous results hold true for $\tau_\UEG(\rho_0,\nu_0)$ and $e_\UEG^{\xc}(\rho_0,\nu_0)$.
\end{theorem}

We see using \cref{ethm} (see below) for the upper bound and the Lieb--Thirring and Lieb--Oxford inequalities for the lower bound (see \cref{ltcoro,lobound} below),  
$$
|e_\UEG(\rho_0,\nu_0)|\le c_\TF \rho_0^{5/3} + c_\LO \rho_0^{4/3} + c'\rho_0\nu_0.
$$
Note that in $e_\UEG(\rho_0,\nu_0)$ only the magnitude of $\vnu_0$ matters due to rotational invariance in the thermodynamic limit (see \cref{etranslemma}(ii)). 
It is clear from \cref{corrthm} that we recover the nonmagnetic case by putting $\nu_0=0$, i.e. $e_\UEG(\rho_0,0)=e_\UEG(\rho_0)$. 
Moreover, \cref{kinweizgauge} and the Lieb--Oxford inequality implies that $e_\UEG(\rho_0,\nu_0)\ge  c \rho_0\nu_0 - c_\LO\rho_0^{4/3}$,
so $e_\UEG(\rho_0,\nu_0)\to +\infty$ as $\nu_0\to\infty$ for every fixed $\rho_0>0$. 

\begin{remark}
Unfortunately, we cannot offer a definition of the \emph{classical} UEG in the magnetic case, akin to \cite{lewin2018statistical} in the nonmagnetic case.
The reason is that $\vjp$ uses the phase of the wavefunction in an essential way, and we have no obvious classical probabilistic interpretation of that.
\end{remark}

\subsection{Upper bound on the kinetic energy functional and representability}

In this section, we describe a technical tool that is needed in the proof of \cref{thermothm}, but it is also of independent interest mainly because 
it improves previous results on the (mixed) representability problem for density-current density pairs (see \cref{repsec}).

If $\rho_\gamma$ and $\vjp_\gamma$ comes from an $H^1$-regular one-particle density matrix $\gamma$, we know from \cref{kinweizgauge} that we necessarily have
that 
$$
\int_{\RR^d} \rho_\gamma \left|\vD_a\frac{\vjp_\gamma}{\rho_\gamma}\right|<\infty.
$$
Note that if $d=3$, the the condition is simply $\int_{\RR^3} \rho_\gamma |\vnu_\gamma|<\infty$, where $\vnu_\gamma=\rot \frac{\vjp_\gamma}{\rho_\gamma}$ is the vorticity.

Using the method of proof of \cite[Theorem 3]{lewin2020local}, we obtain the following.
\begin{theorem}\label{repthm}
Let $q$ be the number of spin states. Fix $\rho\in L^1(\RR^d;\RR_+)$ such that $\grad\sqrt{\rho}\in L^2(\RR^d)^d$. 
Let $\vv\in L^2_\rho(\RR^d)^d$ and suppose that $\vv$ admits a decomposition $\vv=\grad g + \vw$, where $\grad g\in L^2_\rho(\RR^d)$, $\vw\in L^2_\rho(\RR^d)^d$, $\vD\vw\in L^1_\rho(\RR^d)^{d\times d}$.

Then there exists a one-particle density matrix $\gamma$ with finite kinetic energy, such that $\rho_\gamma=\rho$, $\vjp_\gamma=\vjp:=\rho\vv$ and
\begin{equation}\label{repineq}
\begin{aligned}
\Tr(-\lapl\gamma)&\le (1+\kappa_1\epsilon) q^{-2/d} c_\TF \int_{\RR^d} \rho^{1+2/d} + \frac{\kappa_2(1+\sqrt{\epsilon})^2}{\epsilon} \int_{\RR^d} |\grad\sqrt{\rho}|^2 \\
&+\int_{\RR^d}\frac{|\vjp|^2}{\rho} + C_d\int_{\RR^d} \rho|\vD\vw|
\end{aligned}
\end{equation}
for any $\epsilon>0$, and for some universal constants $\kappa_1,\kappa_2>0$. Here, the (dimension-dependent) \emph{Thomas--Fermi constant} $c_\TF$ is given by
$c_\TF:=\frac{d}{d+2}\frac{4\pi^2}{|B_1|^{2/d}}$ and $C_d=1+\frac{d^3}{4}$.
\end{theorem}

The inequality \cref{repineq} is \emph{not} gauge-invariant, to see this, take $\vv(\vx)=(0,x_1,0)$ (Landau gauge) and $\vw(\vx)=\frac{1}{2}\ve_3\times\vx$ (symmetric gauge),
with $g(\vx)=\frac{1}{2}x_1x_2$ and some appropriate $\rho$. 
Then $\vD_a\vv=\vD_a\vw$, $\vD_s\vv\neq 0$ but $\vD_s\vw=0$. Clearly, $\vw=\vv$ and $g=0$ is also a possible choice. The lack of gauge invariance in \cref{repineq} solely comes from the presence of the symmetric derivative $\vD_s\vw$,
which we discuss further below. The utility of the formulation of \cref{repthm} is that with a proper choice of gauge, we can hope to make $\vD_s\vw$ vanish.

 The condition \cref{erikcond} of \cref{erikrep}  is easily seen to imply $\vD\vw\in L^1_\rho(\RR^d)^{d\times d}$ (up to a gauge) via the Cauchy--Bunyakovsky--Schwarz inequality.
We also note that for $\vv=0$, we \emph{exactly} recover the statement (and proof) of \cite[Theorem 3]{lewin2020local}, since that construction actually yields a state $\gamma$ with $\vjp_\gamma=0$. Based on the kinetic energy of the ground state of the noninteracting electron gas, we conjecture that the
optimal constant in 3D is $C_3=1$.

The last term on the r.h.s. of \cref{repineq} has similar role as the von Weizs\"acker term. It can be rewritten as
\begin{align*}
\int_{\RR^d} \rho|\vD\vw|&=\int_{\RR^d} \rho \sqrt{|\vD_s\vw|^2+|\vD_a\vw|^2}=\int_{\RR^d} \rho\sqrt{|\vD_s\vw|^2+|\vD_a\vv|^2},
\end{align*}
where we used the fact $\vD_a \grad p=0$ for any $C^2$ scalar field $p:\RR^d\to\RR$, due to Young's theorem.
Hence
\begin{align*}
\int_{\RR^d} \rho |\vD\vw|\simeq \int_{\RR^d} \rho |\vD_s\vw| + \int_{\RR^d} \rho |\vD_a\vv|,
\end{align*}
in the sense of norm equivalence. In 3D, the second term on the r.h.s. is proportional to $\int_{\RR^3} \rho|\vnu|$, which is gauge-invariant.

The first term on the r.h.s. of the preceding equation (with $\vw = \vw_\gamma$)---which we call the \emph{strain term}---is not necessarily finite for $H^1$-regular $\gamma$. 
However, if we assume that $\gamma$ is $H^2$-regular (and the gauge function has $\vD^2 g\in L^1_\rho(\RR^d)^{d\times d}$), then it is not hard to show that the strain term is finite.
Recall that the eigenstates of the magnetic Schr\"odinger Hamiltonian are in $H^2$.
Although \cref{repthm} does not provide the complete solution to the mixed density-current $N$-representability problem (it also does not provide an $H^2$-regular $\gamma$ without further assumptions, sadly), it seems to provide weak enough assumptions that any practically 
relevant $(\rho,\vjp)$ pair verifies it.
That being said, it is our belief that the presence of the strain term is only an artifact of the method of the proof of both \cref{erikrep} and \cref{repthm}.
It remains a challenge to devise a gauge-invariant analogue of \cref{repineq}.

Motivated by \cref{repthm}, we define the \emph{smaller} set $\wt{\rcdens}\subset\rcdens$, via
$$
\wt{\rcdens}:=\left\{ \begin{aligned}
(\rho,\vjp) : &\;\rho\in\rdens, \, \int_{\RR^d} \frac{|\vjp|^2}{\rho}<\infty, \, \int_{\RR^d} \rho \left|\vD\vw\right|<\infty,\;\text{where}\; \frac{\vjp}{\rho}=\grad g + \vw,\\
&\text{for some}\; \grad g\in L^2_\rho(\RR^d)\; \text{and}\;\vw\in L_\rho^2(\RR^d)^d
\end{aligned}\right\}.
$$

With this notation at hand, we can provide an upper bound on the kinetic energy functional $T(\rho,\vjp)$.

\begin{corollary}\label{tthm}
For any $(\rho,\vjp)\in\wt{\rcdens}$, we have $T(\rho,\vjp)\le \text{r.h.s. of \cref{repineq}}$.
\end{corollary}

Moreover, using the fact \cite{bach1994generalized} that for every one-particle density matrix $\gamma$ there exists a quasi-free state $\vGamma$, such that $\gamma_\vGamma=\gamma$, we obtain the following.
\begin{corollary}\label{ethm}
For any $(\rho,\vjp)\in\wt{\rcdens}$, we have $E(\rho,\vjp)\le \text{r.h.s. of \cref{repineq}}$.
\end{corollary}

In fact, for a quasi-free state $\vGamma$, the Coulomb energy of $\vGamma$ minus $D(\rho_\vGamma)$ is always nonpositive.


\section{Proofs}


The rest of the paper is the devoted to proofs. Since \cite{lewin2020local} explains technical details in a rather complete manner, we will avoid repetition whenever possible.

\subsection{Proof of \cref{repthm}}
\leavevmode

\noindent\textbf{Step 0.} 
First, define the ground-state one-particle density matrix $f_t(\vx-\vy)$ of the noninteracting uniform electron gas of density $t\ge 0$, whose translation-invariant kernel is given by
$$
f_t(\vz)=\int_{|\vk|^2\le \frac{d+2}{d}c_{\TF}t^{2/d}} e^{i \vk\cdot \vz}\,\frac{\dd\vk}{(2\pi)^d}.
$$
It is immediate that $f_t$ is Hermitian and $0\le f_t\le \iden$ in the sense of operators. Also, its density is $f_t(0)=t$. 
Moreover, $\grad_\vx f_t(\vx-\vy)|_{\vy=\vx}=\grad_\vy f_t(\vx-\vy)|_{\vx=\vy}=0$. Furthermore, the kinetic energy density is also constant 
$$
\grad_\vx\cdot \grad_\vy f_t(\vx-\vy)|_{\vy=\vx}=c_{\TF}t^{1+2/d}.
$$

\noindent\textbf{Step 1.} 
As the first derivative of $f_t$ vanishes on the diagonal, this state has zero paramagnetic current density. 
We would like to construct a state $f_t^\vu(\vx-\vy)$ which has a constant paramagnetic current density proportional to a fixed vector $\vu\in\RR^d$.
To achieve this, we simply shift the Fermi sphere $|\vk|^2\le\frac{d+2}{d}c_{\TF}t^{2/d}$ by $\vu$. More precisely, we consider the translation-invariant kernel $f_t^\vu(\vx-\vy)$, where
\begin{align*}
f_t^\vu(\vz)&=\int_{|\vk-\vu|^2\le \frac{d+2}{d}c_{\TF}t^{2/d}} e^{i \vk\cdot \vz}\,\frac{\dd\vk}{(2\pi)^d}=e^{i\vu\cdot \vz} f_t(\vz).
\end{align*}
Clearly, $f_t^\vu$ is Hermitian and $0\le f_t^\vu\le \iden$ in the sense of operators. Its density is $f_t^\vu(0)=t$, but now
\begin{equation}\label{ftvgrad}
\begin{aligned}
\grad_\vx f_t^\vu(\vx-\vy)|_{\vy=\vx} &= i t \vu,\\
\grad_\vy f_t^\vu(\vx-\vy)|_{\vy=\vx} &= -i t \vu.
\end{aligned}
\end{equation}
Therefore, this state has paramagnetic current density 
$$
\vjp(\vx)=\Imag \grad_\vx f_t^\vu(\vx-\vy)|_{\vy=\vx}=t\vu
$$
as desired.
Further, the kinetic energy density is again constant, but has an additional term
$$
\grad_\vy\cdot \grad_\vx f_t^\vu(\vy-\vx)|_{\vy=\vx}= t|\vu|^2 + c_{\TF} t^{1+2/d}.
$$

\noindent\textbf{Step 2.} 
For simplicity, set $q=1$. We are now ready to construct a fermionic one-particle density matrix $\gamma$ with prescribed density $\rho(\vx)$ and paramagnetic current density of the form $\vjp=\rho\vv$. 
Let us first assume that $\vv=\vw$. Following \cite{lewin2020local}, we define\footnote{Note that this state reduces to the cited one by taking $\vv\equiv 0$ and $\theta=\delta_0$.}
$$
\gamma(\vx,\vy)=\int_{\RR^d} \int_0^{+\infty} \sqrt{\theta(\vu-\vv(\vx)) \eta\left(\frac{t}{\rho(\vx)}\right)} f_t^\vu(\vx-\vy) \sqrt{\theta(\vu-\vv(\vy)) \eta\left(\frac{t}{\rho(\vy)}\right)}\, 
\frac{\dd t}{t} \,\dd\vu,
$$
where the weight function $\eta:\RR_+\to\RR_+$ is smooth, and satisfies
$$
\int_0^{+\infty} \eta=1, \quad\text{and}\quad \int_0^{+\infty} \frac{\eta(t)}{t}\,\dd t\le 1,
$$
and $\theta:\RR^d\to\RR_+$ is also smooth and satisfies
\begin{equation}\label{mucond}
\int_{\RR^d} \theta=1, \quad \int_{\RR^d} \vu \theta(\vu)\,\dd\vu=0, \quad\text{and}\quad \int_{\RR^d} \grad\theta(\vu) \,\dd\vu=0.
\end{equation}
These will be determined later in the proof.
Then we have 
$$
\rho_\gamma(\vx)=\int_{\RR^d} \int_0^{+\infty} \theta(\vu-\vv(\vx)) \eta\left(\frac{t}{\rho(\vx)}\right) \, \dd t \,\dd\vu=\rho(\vx),
$$
and
$$
0\le \gamma\le \int_{\RR^d} \int_0^{+\infty} \theta(\vu-\vv(\vx)) \eta\left(\frac{t}{\rho(\vx)}\right) \, \frac{\dd t}{t} \,\dd\vu\le \iden.
$$
Furthermore, we find
\begin{align*}
\vjp_\gamma(\vx)=\Imag \grad_\vx\gamma(\vx,\vy)|_{\vy=\vx}&=\int_{\RR^d} \theta(\vu-\vv(\vx))\vu\,\dd\vu \int_0^{+\infty} \eta\left(\frac{t}{\rho(\vx)}\right)  \, \dd t=\rho(\vx)\vv(\vx)
\end{align*}
as required.
Next, we compute the kinetic energy density,
\begin{align*}
&\grad_\vy\cdot \grad_\vx\gamma(\vx,\vy)|_{\vy=\vx}=\int_{\RR^d} \int_0^{+\infty} \left|\grad \sqrt{\theta(\vu-\vv(\vx)) \eta\left(\frac{t}{\rho(\vx)}\right)}\right|^2\, 
\dd t \,\dd\vu\\
&+\int_{\RR^d} \int_0^{+\infty} \theta(\vu-\vv(\vx)) \eta\left(\frac{t}{\rho(\vx)}\right) (|\vu|^2 + c_{\TF} t^{2/d}) \, \dd t \,\dd\vu\\
&=: \mathrm{(I)} + \mathrm{(II)}
\end{align*}
since the sum of the first order derivatives of $f_t^\vu$ cancel on the diagonal, see \cref{ftvgrad}.
Here, using the first two relations of \cref{mucond}, we have
\begin{align*}
\mathrm{(II)}=\rho(\vx)|\vv(\vx)|^2 + \rho(\vx) \int_{\RR^d} \theta(\vu)|\vu|^2\,\dd\vu + c_\TF \rho(\vx)^{1+2/d} \int_0^{+\infty} \eta(s)s^{2/d}\,\dd s.
\end{align*}
Moreover,
\begin{align*}
\mathrm{(I)}&=\rho(\vx) \int_{\RR^d} \frac{|\vD\vv(\vx)\grad\theta(\vu)|^2}{4\theta(\vu)} \,\dd\vu + |\grad\sqrt{\rho}(\vx)|^2 \int_0^{+\infty} \frac{t^2\eta'(t)^2}{\eta(t)}\,\dd t\\
&+\frac{1}{2} \vD\vv(\vx)\frac{\grad \rho(\vx)}{\rho(\vx)}\cdot  \int_{\RR^d} \grad\theta(\vu) \,\dd\vu \underbrace{\int_0^{+\infty}  t \eta'(t) \,\dd t}_{=1}.
\end{align*}
The last term vanishes because $\int_{\RR^d} \grad\theta(\vu) \,\dd\vu=0$, by hypothesis.
In summary, 
\begin{align*}
&\grad_\vy\cdot \grad_\vx\gamma(\vx,\vy)|_{\vy=\vx}=c_\TF \rho(\vx)^{1+2/d} \int_0^{+\infty} \eta(s)s^{2/d}\,\dd s + |\grad\sqrt{\rho}(\vx)|^2\int_0^{+\infty}\frac{t^2\eta'(t)^2}{\eta(t)}\,\dd t \\
& + \rho(\vx) |\vv(\vx)|^2 + \underbrace{\rho(\vx) \int_{\RR^d} \theta(\vu)|\vu|^2\,\dd\vu + \rho(\vx) \int_{\RR^d} \frac{|\vD\vv(\vx)\grad\theta(\vu)|^2}{4\theta(\vu)} \,\dd\vu.}_{(*)}
\end{align*}

To optimize $\eta$, we simply follow \cite{lewin2020local} because there are no mixed terms containing both $\eta$ and $\theta$.
The last term can be bounded by 
$$
\rho(\vx)|\vD\vv(\vx)|^2 \int_{\RR^d} \frac{|\grad\theta(\vu)|^2}{4\theta(\vu)} \,\dd\vu.
$$
Clearly, it is enough to consider radial $\theta(\vu)=\theta(r)$ for the optimization of the prefactors involving $\theta$, where $r=|\vu|$. After switching to spherical coordinates, we
obtain that with
$$
C(\delta):= d|B_1| \inf_{\substack{\theta\ge 0\\\int_0^{+\infty} r^{d-1}\theta(r)\,\dd r =(d|B_1|)^{-1}\\ \int_0^{+\infty} r^{d+1}\theta(r)\,\dd r=(d|B_1|)^{-1}\delta}} \int_0^{+\infty} \frac{r^{d-1} \theta'(r)^2}{4\theta(r)}\,\dd r,
$$
the bound $(*)\le \delta \rho(\vx) + C(\delta)\rho(\vx)|\vD\vv(\vx)|^2$
holds true.
The function
\begin{equation}\label{thetagauss}
\theta(r)=\left(\frac{2\pi\delta}{d}\right)^{-d/2}  e^{-\frac{d}{2\delta}r^2}
\end{equation}
is admissible, and has
\begin{align*}
C(\delta)\le d|B_1| \int_0^{+\infty} \frac{r^{d-1} \theta'(r)^2}{4\theta(r)}\,\dd r &=\frac{d^3}{4\delta}.
\end{align*}
Hence,
\begin{align*}
(*)&\le \delta \rho(\vx) + \frac{d^3}{4\delta} \rho(\vx)|\vD\vv(\vx)|^2 = C_d \rho(\vx)|\vD\vv(\vx)|,
\end{align*}
with the choice $\delta=|\vD\vv(\vx)|$ and $C_d=1+\frac{d^3}{4}$.

In summary, we end up with
\begin{equation}\label{vwkinbound}
\begin{aligned}
\Tr(-\lapl\gamma)&\le (1+\kappa_1\epsilon) c_\TF \int_{\RR^d} \rho^{1+2/d} + \frac{\kappa_2(1+\sqrt{\epsilon})^2}{\epsilon} \int_{\RR^d} |\grad\sqrt{\rho}|^2  + \int_{\RR^d} \rho|\vv|^2\\
&+C_d\int_{\RR^d} \rho(\vx)|\vD\vv(\vx)| \,\dvx,
\end{aligned}
\end{equation}
which completes the proof for $\vv=\vw$.

\noindent\textbf{Step 3.} 
For the general case $\vv=-\grad g+\vw$, we apply the gauge transformation $g$ on $\gamma$ to get $\wt{\gamma}$. \cref{densgauge} together with \cref{vwkinbound} (where $\vv$ is replaced by $\vw$) gives
\begin{align*}
\Tr(-\lapl\wt{\gamma})&=\Tr(-\lapl\gamma)-2\int_{\RR^d} \rho\vw \cdot \grad g + \int_{\RR^d} \rho|\grad g|^2\\
&\le (1+\kappa_1\epsilon) c_\TF \int_{\RR^d} \rho^{1+2/d} + \frac{\kappa_2(1+\sqrt{\epsilon})^2}{\epsilon} \int_{\RR^d} |\grad\sqrt{\rho}|^2  + \int_{\RR^d} \rho|\vv|^2\\
&+ C_d\int_{\RR^d} \rho|\vD\vw|,
\end{align*}
which finishes the proof of \cref{repthm}.

\subsection{Further bounds}

We start by recalling a simple version of the magnetic Lieb--Thirring inequality, which can be used to derive a gauge invariant bound.
Let $\vA\in \Lloc^2(\RR^3)^3$ and denote the corresponding \emph{magnetic Sobolev space} as $H^1_\vA(\RR^3)$, see \cite{lieb2001analysis}.

\begin{theorem}[Magnetic Lieb--Thirring kinetic energy inequality]\cite{lieb2010stability}
Fix $\vA\in\Lloc^2(\RR^3)$. For any fermionic one-particle reduced density matrix $\gamma$ in $H^1_\vA(\RR^3)$, the bound
$$
\Tr(-i\grad + \vA)^2\gamma \ge C_\LT \int_{\RR^3} \rho_\gamma^{5/3} 
$$
holds true for some universal constant $C_\LT>0$.
\end{theorem}

Using this, the usual Lieb--Thirring kinetic energy inequality can be cast in a gauge-invariant form.
\begin{corollary}[Gauge-invariant Lieb--Thirring inequality]\label{ltcoro}
For any fermionic one-particle reduced density matrix $\gamma$ with finite kinetic energy, the bound
$$
\Tr(-\lapl\gamma) \ge c_\LT \int_{\RR^3} \rho_\gamma^{5/3} + \int_{\RR^3} \frac{|\vjp_\gamma|^2}{\rho_\gamma}
$$
holds true for some universal constant $c_\LT>0$. 
\end{corollary}
\begin{proof}
For any $\epsilon>0$, we set $\vA_\epsilon=-\frac{\vjp_\gamma}{\rho_\gamma+\epsilon\sqrt{\rho_\gamma}}$. Then $\vA_\epsilon\in\Lloc^2(\RR^3)$ since $\int_{\RR^3} \frac{|\vjp_\gamma|^2}{\rho_\gamma}<\infty$. Note that 
\begin{align*}
C_\LT \int_{\RR^3} \rho_\gamma^{5/3}&\le \Tr(-i\grad + \vA_\epsilon)^2\gamma = \Tr(-\lapl\gamma) +2\int_{\RR^3} \vA_\epsilon \cdot \vjp_\gamma + \int_{\RR^3} \rho_\gamma |\vA_\epsilon|^2\\
&=\Tr(-\lapl\gamma) - \int_{\RR^3} |\vjp_\gamma|^2 \frac{\rho_\gamma+2\epsilon\sqrt{\rho_\gamma}}{(\rho_\gamma+\epsilon\sqrt{\rho_\gamma})^2}.
\end{align*}
An application of Fatou's lemma finishes the proof.
\end{proof}

For any $\vGamma\in\dm$ let us introduce the notations
$$
\TC(\vGamma)=\sum_{n\ge 1}\sum_{1\le j\le n} \Tr_{\Hil^n}(-\lapl_{\vx_j}\Gamma_n), \quad \Cc(\vGamma)=\sum_{n\ge 1}\sum_{1\le j<k\le n} \Tr_{\Hil^n}\left( \frac{1}{|\vx_j-\vx_k|} \Gamma_n \right)
$$
for the \emph{kinetic-}, and the \emph{Coulomb energy of the state} $\vGamma$. Of course, $\TC(\vGamma)=\Tr(-\lapl\gamma_\vGamma)$, so \cref{ltcoro} provides a gauge-invariant lower bound on $\TC(\vGamma)$. 
The following famous estimate on the indirect Coulomb energy will be used repeatedly in the sequel (which is clearly gauge-invariant).  
\begin{theorem}[Lieb--Oxford inequality]\label{lobound}\cite{lieb1979lower,lieb1981improved,chan1999optimized,lewin2022improved}
There exists a universal constant $c_\LO>0$ such that for any Fock space density matrix $\vGamma$ (not necessarily fermionic or $H^1$-regular), the bound
\begin{align*}
\Cc(\vGamma)-D(\rho_\vGamma)\ge - c_\LO \int_{\RR^3} \rho_\vGamma^{4/3}
\end{align*}
holds true.
\end{theorem}

\subsection{Transformation properties}
In this section, we discuss the behavior of the density functionals under various transformations of the density-current density pair. 

First, we consider general affine transformations of states. Let $\vGamma=\Gamma_0\oplus\Gamma_1\oplus\ldots\in\dm$ and for any $\vM\in\RR^{3\times 3}$ nonsingular and $\va\in\RR^3$, set $\vT(\vx)=\vM\vx+\va$, and
define the transformed state $\Upsilon_\vT\vGamma$ by setting 
$$
(\Upsilon_\vT\Gamma_n)(\vX;\vY)=(\det\vM)^n\Gamma_n(\vT(\vx_1),\ldots,\vT(\vx_n);\vT(\vy_1),\ldots,\vT(\vy_n)).
$$
Then the reduced one-particle density matrix of $\Upsilon_\vT\vGamma$ is given by $\gamma_{\Upsilon_\vT\vGamma}(\vx,\vy)=(\det\vM)\gamma_\vGamma(\vT(\vx),\vT(\vy))$, where $\gamma_\vGamma$ denotes the one-particle density matrix of $\vGamma$.
We have $\rho_{\Upsilon_\vT\vGamma}(\vx)=(\det\vM) \rho_{\vGamma}(\vT(\vx))$ and $\vjp_{\Upsilon_\vT\vGamma}(\vx)=(\det\vM)\vM^\top \vjp_{\vGamma}(\vT(\vx))$. Further, $\gamma^{(2)}_{\Upsilon_\vT\vGamma}(\vx,\vy;\vx,\vy)=(\det\vM)^2\gamma^{(2)}_\vGamma(\vT(\vx),\vT(\vy);\vT(\vx),\vT(\vy))$.
We immediately obtain that for any $\vGamma\in\dm$ the following relations hold true:
\begin{align}\label{Ttrans}
\TC(\Upsilon_\vT\vGamma)&=\int_{\RR^3} \Tr_{\RR^3}\big( \vM\vM^\top(\grad_\vx\otimes\grad_\vy)\gamma_\vGamma(\vx,\vy)|_{\vy=\vx} \big) \,\dd\vx,\\
\label{Ctrans}
\Cc(\Upsilon_\vT\vGamma)&=\iint_{\RR^3\times\RR^3} \frac{\gamma^{(2)}_\vGamma(\vx,\vy;\vx,\vy)}{|\vM^{-1}(\vx-\vy)|} \,\dd\vx\dd\vy.
\end{align}
Motivated by the above considerations, for any $(\rho,\vjp)\in\rcdens$, we define 
\begin{equation}\label{upsdef}
\boxed{\Upsilon_\vT(\rho,\vjp)(\vx):=(\det\vM)\big(\rho(\vT(\vx)),\vM^\top\vjp(\vT(\vx))\big)}
\end{equation}
as a slight abuse of notation. 
Using \cref{Ttrans} and \cref{Ctrans}, and introducing the the \emph{gauge-corrected kinetic energy tensor} 
$$
\ol{\vtau}_{\gamma}=\vtau_{\gamma} - \frac{\vjp_{\gamma}\otimes\vjp_{\gamma}}{\rho_\gamma},
$$
the transformed gauge-corrected indirect energy may be expressed as
\begin{equation}\label{transrep}
\ol{E}(\Upsilon_\vM(\rho,\vjp))=\inf_{\substack{\vGamma\in\dm\\\rho_\vGamma=\rho\\\vjp_\vGamma=\vjp}}\Bigg[ \int_{\RR^3} \Tr_{\RR^3}\left( \vM\vM^\top \ol{\vtau}_{\gamma_\vGamma}(\vx) \right) \,\dd\vx  +\Cc(\Upsilon_\vM\vGamma)\Bigg]-D(\Upsilon_\vM\rho),
\end{equation}
which follows by a simple reparametrization of the infimum, using the fact that the map $\vGamma\mapsto\Upsilon_\vM\vGamma$ is invertible.

In ordinary DFT, the Levy--Lieb functional is invariant under isometries. In CDFT, we have the following
transformation rule.
\begin{lemma}\label{isolemma}
Let $\vT(\vx)=\vR\vx + \va$ be an isometry for some $\vR\in\mathrm{O}(3)$ and $\va\in\RR^3$. Then
$$
\Upsilon_\vT(\rho,\vjp)=\big(\rho\circ\vT, \vR^\top(\vjp\circ\vT)\big)
$$
and
$$
T(\Upsilon_\vT(\rho,\vjp))=T(\rho,\vjp), \quad\text{and}\quad F(\Upsilon_\vT(\rho,\vjp))=F(\rho,\vjp).
$$
\end{lemma}
The proof follows from relations analogous to \cref{transrep}.
As an application, we deduce how the indirect energy $E$ and $e_{\Omega,\eta_\delta}(\rho_0,\vnu_0)$, $\tau_{\Omega,\eta_\delta}(\rho_0,\vnu_0)$ and $e^{\xc}_{\Omega,\eta_\delta}(\rho_0,\vnu_0)$ 
transform under isometries. The proof uses the gauge transformation rule \cref{vrgauge} in an essential way.

\begin{lemma}\label{etranslemma}
Fix $\rho_0\in\RR_+$ and $\vnu_0\in\RR^3$.
Let $\vT(\vx)=\vR\vx + \va$ be an isometry for some $\vR\in\mathrm{O}(3)$ and $\va\in\RR^3$. Let $\Omega\subset\RR^3$ be a bounded domain with barycenter 0, 
and let $\eta$ be a regularization function.
\begin{itemize}
\item[(i)] The following formula holds true:
\begin{multline*}
E\big((\iden_{\vT^{-1}(\Omega)}*(\eta_\delta\circ\vR))(\rho_0,\tfrac{1}{2}\rho_0\vnu_0\times\cdot)\big)=E\big( (\iden_{\Omega}*\eta_\delta)(\rho_0,\tfrac{1}{2}\rho_0\vR\vnu_0\times\cdot)\big)\\
 + \frac{\rho_0}{4}|\vR\vnu_0\times \va|^2 \int_{\RR^3} (\iden_{\Omega}*\eta_\delta).
\end{multline*}
Further, the same relation holds with $E$ replaced by $T$. 
\item[(ii)] The relation $e_{\vT^{-1}(\Omega),\eta_\delta\circ\vR}(\rho_0,\vnu_0)= e_{\Omega,\eta_\delta}(\rho_0,\vR\vnu_0)$
holds true. Moreover, the formula is also valid if $e$ is replaced by $\tau$ and $e^{\xc}$.
\end{itemize}
\end{lemma}

\begin{proof}
First, we show (i). Clearly, $\vT^{-1}(\vx)=\vR^\top(\vx-\va)$. Also, $D(((\iden_{\Omega}*\eta_\delta) \circ \vT)\rho_0)=D((\iden_{\Omega}*\eta_\delta)\rho_0)$. We may write the l.h.s. using \cref{isolemma} as
\begin{align*}
&E\big(((\iden_{\Omega}*\eta_\delta) \circ \vT) (\rho_0,\tfrac{1}{2}\rho_0\vnu_0 \times\cdot)\big)\\
&=E\Big(\vx\mapsto (\iden_{\Omega}*\eta_\delta)(\vx)\big(\rho_0,\tfrac{1}{2}\rho_0\vR(\vnu_0\times \vT^{-1}(\vx))\big)\Big)\\
&=E\Big(\vx\mapsto (\iden_{\Omega}*\eta_\delta)(\vx) \big(\rho_0,\tfrac{1}{2}\rho_0 \vR(\vnu_0\times \vR^{\top}\vx)+\tfrac{1}{2}\rho_0\vR(\vnu_0\times\vR^{\top}\va)) \big)\Big)\\
&=E\Big(\vx\mapsto (\iden_{\Omega}*\eta_\delta)(\vx) \big(\rho_0,\tfrac{1}{2}\rho_0 \vR\vnu_0\times \vx+\tfrac{1}{2}\rho_0 \vR\vnu_0\times\va \big)\Big)
\end{align*}
where we used the identity $\vR(\vu\times\vR^{\top}\vv)=\vR\vu\times\vv$.
Using the gauge transformation rule \cref{vrgauge}, we obtain
\begin{align*}
&E\big(((\iden_{\Omega}*\eta_\delta) \circ \vT) (\rho_0,\tfrac{1}{2}\rho_0\vnu_0\times\cdot)\big)=E\big((\iden_{\Omega}*\eta_\delta) (\rho_0,\tfrac{1}{2}\rho_0\vR\vnu_0\times\cdot)\big)\\
&-\frac{\rho_0}{2}\int_{\RR^3} (\iden_{\Omega}*\eta_\delta)(\vx) (\vR\vnu_0\times \vx)\cdot (\vR\vnu_0\times\va)\,\dd\vx + \frac{\rho_0}{4}|\vR\vnu_0\times\va|^2\int_{\RR^3} (\iden_{\Omega}*\eta_\delta)(\vx) \,\dd\vx\\
&= E\big( (\iden_{\Omega}*\eta_\delta)(\rho_0,\tfrac{1}{2}\rho_0 \vR\vnu_0\times\cdot)\big) + \frac{\rho_0}{4}|\vR\vnu_0\times\va|^2\int_{\RR^3} (\iden_{\Omega}*\eta_\delta),
\end{align*}
as stated. 
Here, the cross term disappeared because $\int_{\RR^3} \vx(\iden_{\Omega}*\eta_\delta)(\vx)\,\dd\vx=0$, since the smeared set also has barycenter 0, due to $\int_{\RR^3}\vx\eta(\vx)\,\dd\vx=0$.

For part (ii), we may write based on part (i) that
\begin{multline*}
e_{\vT^{-1}(\Omega),\eta_\delta\circ\vR}(\rho_0,\vnu_0)= \frac{1}{|\Omega|}\Bigg[ E\big( (\iden_{\Omega}*\eta_\delta)(\rho_0,\tfrac{1}{2}\rho_0\vR\vnu_0\times\cdot)\big) + \frac{\rho_0}{4}|\vR\vnu_0\times \va|^2 \int_{\RR^3} (\iden_{\Omega}*\eta_\delta)\\
-\frac{\rho_0}{4} \int_{\RR^3} (\iden_{\Omega}*\eta_\delta)(\vx) |\vnu_0\times(\vR^{-1}(\vx-\va))|^2 \,\dd\vx \Bigg]=e_{\Omega,\eta_\delta}(\rho_0,\vR\vnu_0),
\end{multline*}
where the mixed term in last term vanished as before.
\end{proof}

Choosing $\vM=\lambda^{1/3}\iden$ in \cref{transrep} and setting 
$$
\delta_\lambda(\rho,\vjp)(\vx):=\Upsilon_{\vM}(\rho,\vjp)(\vx)=(\lambda\rho(\lambda^{1/3}\vx),\lambda^{4/3}\vjp(\lambda^{1/3}\vx)),
$$
 we deduce the following (the lower bound uses the Lieb--Oxford inequality as well).
\begin{theorem}
For every $(\rho,\vjp)\in\rcdens$ and $0\le\lambda\le 1$, we have
\begin{equation}\label{Escale}
\ol{E}(\delta_\lambda(\rho,\vjp))\le \lambda^{1/3} \ol{E}(\rho,\vjp)
\end{equation}
and
\begin{equation}\label{Escalelo}
\begin{aligned}
\ol{E}(\delta_\lambda(\rho,\vjp)) \ge \lambda^{2/3} \ol{E}(\rho,\vjp) + (\lambda^{2/3}-\lambda^{1/3}) c_\LO \int_{\RR^3} \rho^{4/3}.
\end{aligned}
\end{equation}
\end{theorem}

\subsection{Decoupling lower bound}

Let $C_1:=(-\frac{1}{2},\frac{1}{2})^3$ denote the origin-centered unit cube, and $\tetra_j=\vT_j\tetra$ ($j=1,\ldots,24$) be its tiling with 24 congruent tetrahedra, where
$\vT_j\vx=\vR_j\vx-\vz_j$ for some $\vR_j\in \SO(3)$ and $\vz_j\in C_1$. Here, $\tetra$ is a ``reference'' tetrahedron with barycenter 0.
Then with $C_\ell=\ell C_1$, we write 
$$
\RR^3=\bigcup_{\vz\in\ZZ^3}(C_\ell+\ell\vz)=\bigcup_{\vz\in\ZZ^3} \bigcup_{j=1}^{24} ( \ell \vT_j\tetra + \ell\vz),
$$
where $\ell>0$ the scale of the tiling. The reason behind the use of this particular tetrahedral tiling is that the Graf--Schenker inequality
is employed for the Coulomb interaction, see \cite{lewin2018statistical,lewin2020local} and the original \cite{graf1995molecular,graf1995electrostatic} for details.

The corresponding partition of unity is given by 
\begin{equation}\label{pouind}
\sum_{\vz\in\ZZ^3} \sum_{j=1}^{24} \iden_{\ell \vT_j\tetra}(\cdot-\ell\vz)\equiv 1\;\text{(a.e.)}.
\end{equation}

The lower bound is based on well-known localization techniques, and for that, smoothing is needed.
Let $0\le \eta\in C_c^\infty(\RR^3)$ be such that $\int_{\RR^3} \vx\eta(\vx)\,\dd\vx=0$ with $\supp\eta\subset B_1$ and $\int_{\RR^3}\eta=1$. Define
$$
\eta_\delta(\vx)=\left(\frac{10}{\delta}\right)^3\eta\left(\frac{10\vx}{\delta}\right).
$$
Then $\supp\eta_\delta\subset B_{\frac{\delta}{10}}$ and $\int_{\RR^3} \eta_\delta=1$. 
The functions $\iden_{\ell\tetra_j} * \eta_\delta$ form a partition of unity in the following sense:
\begin{align*}
\sum_{\vz\in\ZZ^3} \sum_{j=1}^{24} (\iden_{\ell\tetra_j} * \eta_\delta)(\cdot - \ell\vz)&=\sum_{\vz\in\ZZ^3} \sum_{j=1}^{24} \int_{\RR^3} \iden_{\ell\tetra_j}(\cdot-\vx-\ell\vz) \eta_\delta(\vx)\,\dd\vx 
=\int_{\RR^3} \eta_\delta \equiv 1
\end{align*}
by \cref{pouind}. The following facts will be used repeatedly. 
Fix $\vx\in\RR^3$, if $\supp\eta_\delta(\vx-\cdot)\cap \partial(\ell\tetra_j)=\emptyset$, then clearly either $\supp\eta_\delta(\vx-\cdot)\subset \ell\tetra_j$ or else 
$\supp\eta_\delta(\vx-\cdot)\subset (\ell\tetra_j)^c$. In the former case, $(\iden_{\ell\tetra_j} * \eta_\delta)(\vx)=1$ and in the latter, $(\iden_{\ell\tetra_j} * \eta_\delta)(\vx)=0$. Since $\supp\eta_\delta(\vx-\cdot)\subset B_{\frac{\delta}{10}}(\vx)$, we obtain that 
\begin{equation}\label{xicase}
(\iden_{\ell\tetra_j} * \eta_\delta)(\vx)\equiv \begin{cases}
1, & \text{if}\quad B_{\frac{\delta}{10}}(\vx)\subset\ell\tetra_j \\
0, & \text{if}\quad B_{\frac{\delta}{10}}(\vx)\subset(\ell\tetra_j)^c
\end{cases}
\end{equation}
It follows that $\iden_{\ell\tetra_j} * \eta_\delta\equiv 1$ on $\ell\tetra_j\setminus (\partial(\ell\tetra_j)+B_{\frac{\delta}{10}})$.
More conveniently, $\iden_{\ell\tetra_j} * \eta_\delta\equiv 1$ on $(\ell-\delta)\tetra_j$, as this set is clearly contained in the former.

The following trivial observation shows that $\vjp$ behaves the same way as $\rho$ with respect to localization.
Let $\gamma$ be a one-particle density matrix and fix $0\le\chi\le 1$, where $\chi\in C^1(\RR^3)$, and define
$\gamma|_\chi=\chi\gamma\chi$.
Then $\rho_{\gamma|_\chi}=\chi^2\rho_\gamma$ and $\vjp_{\gamma|_\chi}=\chi^2 \vjp_\gamma$, where we used the fact that $\chi$ real-valued.

This implies that we may proceed \emph{exactly} as in \cite{lewin2018statistical,lewin2020local}. We do not, however, discard the kinetic energy term
$\frac{C\delta}{\ell} T(\ldots)$ for the bound on $E(\rho,\vjp)$, as it will be required to cancel the gauge terms in \cref{tetrathermo}.
Consequently, we have

\begin{theorem}\label{lboundthm}
There is a constant $C>0$ such that for any $\delta>0$ such that $0<\delta\ell<1/C$ and $(\rho,\vjp)\in\wt{\rcdens}$ the bound
\begin{multline*}
E(\rho,\vjp)\ge \int_{\SO(3)}\dd\vR  \dashint_{C_{\ell}}\dd\vtau \sum_{\vz\in\ZZ^3} \sum_{j=1}^{24} \times \\
\left[ \left(1-\frac{C\delta}{\ell}\right) E\bigl((\iden_{\ell\tetra_j} * \eta_\delta)(\vR\cdot - \ell\vz - \vtau)(\rho,\vjp)\bigr)+\frac{C\delta}{\ell} T\bigl((\iden_{\ell\tetra_j} * \eta_\delta)(\vR\cdot - \ell\vz - \vtau)(\rho,\vjp)\bigr)\right] \\ 
-\frac{C}{\ell} \int_{\RR^3} ((1+\delta^{-1})\rho + \delta^3\rho^2 )
\end{multline*}
holds. Furthermore, for any $\delta>0$ and $\ell>0$ the bound
\begin{align*}
T(\rho,\vjp)\ge  \dashint_{C_{\ell}}\dd\vtau \sum_{\vz\in\ZZ^3} \sum_{j=1}^{24} T\bigl((\iden_{\ell\tetra_j} * \eta_\delta)(\vR\cdot - \ell\vz - \vtau)(\rho,\vjp)\bigr) - \frac{C}{\delta\ell}\int_{\RR^3} \rho
\end{align*}
also holds true.
\end{theorem}

\subsection{Decoupling upper bound}

First, \cref{pouind} is regularized as follows.
For any $\delta\le\frac{\ell}{2}$, consider the functions $(1-\delta/\ell)^{-3} \iden_{\ell \vT_j(1-\frac{\delta}{\ell})\tetra} * \eta_\delta$.
Here, $\ell \vT_j(1-\frac{\delta}{\ell})\tetra=(\ell-\delta)\vR_j\tetra + \ell\vz_j$ and
$$
\int_{\RR^3} (1-\delta/\ell)^{-3} \iden_{\ell \vT_j(1-\frac{\delta}{\ell})\tetra} * \eta_\delta=|\ell\tetra|.
$$
To obtain a true partition of unity, these functions are averaged over the cubic tiling:
\begin{equation}\label{pouavg}
\dashint_{C_\ell} \sum_{\vz\in\ZZ^3}\sum_{j=1}^{24}  (1-\delta/\ell)^{-3} (\iden_{\ell \vT_j(1-\frac{\delta}{\ell})\tetra} * \eta_\delta)(\cdot - \vtau-\ell\vz)\,\dd\vtau \equiv 1\;\text{(a.e.)}.
\end{equation}
The following result says that the (grand-canonical) Vignale--Rasolt functional decouples the same way as the Levy--Lieb functional.
\begin{proposition}
For any $(\rho,\vjp)\in\wt{\rcdens}$, the bounds
\begin{align*}
&F(\rho,\vjp)\le \dashint_{C_{\ell}}\dd\vtau \sum_{\vz\in\ZZ^3} \sum_{j=1}^{24} F\bigl((1-\delta/\ell)^{-3} (\iden_{\ell \vT_j(1-\frac{\delta}{\ell})\tetra} * \eta_\delta)(\cdot - \vtau-\ell\vz)(\rho,\vjp)\bigr)\\
&+\dashint_{C_\ell}\dd\vtau \sum_{\substack{\vz\neq \vz'\\j\neq j'}}(1-\delta/\ell)^{-6} D\bigl( (\iden_{\ell \vT_j(1-\frac{\delta}{\ell})\tetra} * \eta_\delta)(\cdot - \vtau-\ell\vz)\rho, (\iden_{\ell \vT_{j'}(1-\frac{\delta}{\ell})\tetra} * \eta_\delta)(\cdot - \vtau-\ell\vz')\rho\bigr)
\end{align*}
and
\begin{align*}
T(\rho,\vjp)&\le \dashint_{C_{\ell}}\dd\vtau \sum_{\vz\in\ZZ^3} \sum_{j=1}^{24} T\bigl((1-\delta/\ell)^{-3} (\iden_{\ell \vT_j(1-\frac{\delta}{\ell})\tetra} * \eta_\delta)(\cdot - \vtau-\ell\vz)(\rho,\vjp)\bigr)
\end{align*}
hold true.
\end{proposition}
\begin{proof}
The partition of unity \cref{pouavg} allows us to write our density-current pair as
\begin{equation}\label{rhovjppou}
(\rho,\vjp)=\dashint_{C_\ell} \sum_{\vz\in\ZZ^3}\sum_{j=1}^{24} (\rho_{\vtau,\vz,j},\vjp_{\vtau,\vz,j}) \,\dd\vtau,
\end{equation}
where we have set
$$
(\rho_{\vtau,\vz,j},\vjp_{\vtau,\vz,j})= (1-\delta/\ell)^{-3} (\iden_{\ell \vT_j(1-\frac{\delta}{\ell})\tetra} * \eta_\delta)(\cdot - \vtau-\ell\vz)(\rho,\vjp)
$$
for any $\vtau\in C_\ell$, $\vz\in\ZZ^3$ and $j=1,\ldots,24$. 
We now claim that the density-current pairs $(\rho_{\vtau,\vz,j},\vjp_{\vtau,\vz,j})$ are
 representable, i.e. $(\rho_{\vtau,\vz,j},\vjp_{\vtau,\vz,j})\in\wt{\rcdens}$. In fact, it is clear that $\rho_{\vtau,\vz,j}\in\rdens$, and with $\vv=\frac{\vjp}{\rho}$ and 
$$
\vv_{\vtau,\vz,j}:=\frac{\vjp_{\vtau,\vz,j}}{\rho_{\vtau,\vz,j}}\iden_{\supp\rho_{\vtau,\vz,j}}=\vv\iden_{\supp\rho_{\vtau,\vz,j}}
$$
there holds $\int_{\RR^3} \rho_{\vtau,\vz,j}|\vv_{\vtau,\vz,j}|^2 \le \int_{\RR^3} \rho|\vv|^2 <\infty$. 
Also, $\int_{\RR^3} \rho_{\vtau,\vz,j} |\vD\vv_{\vtau,\vz,j}|\le \int_{\RR^3} \rho |\vD\vv|<\infty$. 

The rest of the proof is analogous to that of \cite[Proposition 1]{lewin2020local}. We note that in the construction of the trial state \cref{vropt} is used.
\end{proof}

Next, we state a similar bound for the indirect energy defined in \cref{inddef}.
\begin{theorem}\label{eneup}
For any $(\rho,\vjp)\in\wt{\rcdens}$, $0<\delta<\frac{\ell}{2}$ and $0<\alpha\le\frac{1}{2}$, the bound
\begin{multline*}
E(\rho,\vjp)\\
\le\dashint_{1-\alpha}^{1+\alpha}\frac{\dd t}{t^4} \int_{\SO(3)}\dd\vR  \dashint_{C_{t\ell}}\dd\vtau \sum_{\vz\in\ZZ^3} \sum_{j=1}^{24}
E\bigl((1-\delta/\ell)^{-3} (\iden_{t\ell \vT_j(1-\frac{\delta}{\ell})\tetra} * \eta_{t\delta})(\vR\cdot - t\ell\vz - \vtau)(\rho,\vjp)\bigr)\\
  + C\delta^2\log(\alpha^{-1}) \int_{\RR^3}\rho^2
\end{multline*}
holds for a universal constant $C$. Here, $\dashint_{1-\alpha}^{1+\alpha}\frac{\dd t}{t^4}:= \left(\int_{1-\alpha}^{1+\alpha}\frac{\dd t}{t^4}\right)^{-1} \int_{1-\alpha}^{1+\alpha}\frac{\dd t}{t^4}$.
\end{theorem}
\begin{proof}
We have
\begin{equation}\label{eneupproof}
\begin{aligned}
E(\rho,\vjp)&\le \dashint_{C_\ell}\dd\vtau \sum_{\vz\in\ZZ^3}\sum_{j=1}^{24} E(\rho_{\vtau,\vz,j},\vjp_{\vtau,\vz,j}) +  \dashint_{C_\ell} D(\rho_{\vGamma_\vtau}-\rho)\,\dd\vtau.
\end{aligned}
\end{equation}
According to the proof of \cite[Proposition 1]{lewin2020local} (and Lemma 1 ibid.), replacing $\ell\to t\ell$, $\delta\to t\delta$ and $\rho\to\rho\circ\vR$ (and $\vjp\to\vR^{\top}(\vjp\circ\vR)$), and averaging, the preceding quantity may be bounded as follows
$$
\dashint_{1-\alpha}^{1+\alpha}\frac{\dd t}{t^4} \int_{\SO(3)}\dd\vR  \dashint_{C_{t\ell}}D((\rho_{\vGamma_\vtau}-\rho)\circ\vR)\,\dd\vtau \le C\delta^2\log(\alpha^{-1}) \int_{\RR^3}\rho^2,
$$
for $0< \delta < \ell/2$. 
By applying the same replacements and averagings in \cref{eneupproof}, we may use the preceding estimate and we end up with
\begin{align*}
&E(\rho\circ\vR,\vR^{\top}(\vjp\circ\vR)) \\
&\le \dashint_{1-\alpha}^{1+\alpha}\frac{\dd t}{t^4} \int_{\SO(3)}\dd\vR  \dashint_{C_{t\ell}}\dd\vtau \sum_{\vz\in\ZZ^3}\sum_{j=1}^{24} E((1-\delta/\ell)^{-3} (\iden_{t\ell \vT_j(1-\frac{\delta}{\ell})\tetra} * \eta_{t\delta})(\cdot - t\ell\vz - \vtau) (\rho\circ\vR,\vR^{\top}(\vjp\circ\vR))\\
&+ C\delta^2\log(\alpha^{-1}) \int_{\RR^3}\rho^2
\end{align*}
Recalling \cref{isolemma} finishes the proof.
\end{proof}

\subsection{Energy per volume in a tetrahedron}

\begin{figure}
\centering
\begin{tikzpicture}[scale=4.5]

\newcommand\gs{0.13};
	
	\draw[line width=1mm,red, fill=red, fill opacity=0.2] (-0.5+0*\gs,-0.5+1*\gs) -- (-0.5+1*\gs,-0.5+0*\gs)  -- (-0.5+10*\gs,-0.5+0*\gs) 
	-- (-0.5+10*\gs,-0.5+2*\gs) -- (-0.5+2*\gs,-0.5+10*\gs) -- (-0.5+0*\gs,-0.5+10*\gs)
	-- cycle;
		
	\draw[line width=1mm,white, fill=white] (-0.5+2*\gs,-0.5+2*\gs) -- (-0.5+2*\gs,-0.5+7*\gs) -- (-0.5+7*\gs,-0.5+2*\gs)-- cycle;
	\draw[line width=1mm,blue, opacity=1]  (-0.5+2*\gs,-0.5+2*\gs) -- (-0.5+2*\gs,-0.5+7*\gs) -- (-0.5+7*\gs,-0.5+2*\gs)-- cycle;
	\draw[line width=0,blue, fill=blue, fill opacity=0.4] (-0.5+2*\gs,-0.5+2*\gs) -- (-0.5+2*\gs,-0.5+7*\gs) -- (-0.5+7*\gs,-0.5+2*\gs)-- cycle;

	\draw[gray,step=\gs,xshift=-0.5cm, yshift=-0.5cm] (0,0) grid (\gs*10,\gs*10);

	\foreach \i in {0,...,10} {
	    \draw[gray] (-0.5+\gs*\i,-0.5) -- (-0.5,-0.5+\gs*\i);
	    \draw[gray] (0.8-\gs*\i,0.8) -- (0.8,0.8-\gs*\i);
		}
	
    \draw[thick,dashed] (-1/3,-1/3)--(2/3,-1/3)--(-1/3,2/3)--cycle;

	\draw[thick] (-1/3+0.1,-1/3+0.1)--(-1/3+0.1,2/3-0.25)--(2/3-0.25,-1/3+0.1)--cycle;
	
	\draw[thick] ([shift=(180:0.1)] -1/3,-1/3) arc (180:270:0.1);
	\draw[thick] ([shift=(-90:0.1)] 2/3,-1/3) arc (-90:45:0.1);
	\draw[thick] ([shift=(45:0.1)] -1/3,2/3) arc (45:180:0.1);

	\draw[thick] (-1/3-0.1,-1/3)--(-1/3-0.1,2/3);
	\draw[thick] (-1/3+0.1*0.707,2/3+0.1*0.707)--(2/3+0.1*0.707,-1/3+0.1*0.707);
	\draw[thick] (-1/3,-1/3-0.1)--(2/3,-1/3-0.1);
		
	\draw[thick, >=latex, ->] (0.42,0.32) -- (0.227,0.227);
    \node at (0.65,0.35) {$\supp(\iden_{\ell'\tetra}*\eta_{\delta'})$};
	
	\draw[thick, >=latex, ->] (0.42,0.52) -- (0.07,0.11);
    \node at (0.45,0.55) {$(\ell'-\delta')\tetra$};

\end{tikzpicture}
\caption{The tetrahedra corresponding to $\vR=\iden$ and $\vtau=0$ included in the sets $J_0$ and $J\setminus J_0$ are depicted in blue and red, respectively. 
The tiling is composed of translates of $\ell\tetra_j$.}\label{tilingdecomp}
\end{figure}

We now prove that the thermodynamic limit for the indirect-, and the kinetic energy density of
the uniform electron gas defined in \cref{uegsec} exists when the domain is taken to be a tetrahedron. 


\begin{theorem}[Convergence rate for tetrahedra]\label{tetrathermo}
Fix $\rho_0\in\RR_+$ and $\vnu_0\in\RR^3$. Let $\eta$ and $\wt{\eta}$ be regularization functions. Let $\ul{e}_\UEG(\rho_0,\nu_0)$ and $\ol{e}_\UEG(\rho_0,\nu_0)$ denote the liminf and the limsup of $e_{\ell\tetra,\eta_\delta}(\rho_0,\vnu_0)$
as $\ell\delta\to\infty$ and $\delta/\ell\to 0$. 
\begin{itemize}
\item[(i)] For $\ell/\delta$ sufficiently large, we have
$$
\int_{\SO(3)} e_{\ell\tetra,\wt{\eta}_\delta}(\rho_0,\vR\vnu_0)\,\dd\vR\le \ul{e}_\UEG(\rho_0,\nu_0) + \frac{C\rho_0}{\ell}\Big( 1+\delta^{-1}+\delta^3\rho_0 + \delta \rho_0^{2/3} + \delta\nu_0\Big)
$$
\item[(ii)] We have for all $0<\alpha\le\tfrac{1}{2}$, 
$$
\dashint_{1-\alpha}^{1+\alpha}\int_{\SO(3)} e_{t\ell\tetra,\wt{\eta}_{t\delta}}(\rho_0,\vR\vnu_0)\,\dd\vR \frac{\dd t}{t^4}\ge \ol{e}_\UEG(\rho_0,\nu_0) -
\frac{C\rho_0}{\ell}\big(\delta\rho_0^{2/3} + \ell^{-1}\big) - C\delta^2\log(\alpha^{-1}) \rho_0^2
$$
\item[(iii)] If $\ell\rho_0^{1/3}\ge C$, then 
\begin{align*}
e_{\ell\tetra,\wt{\eta}_\delta}(\rho_0,\vnu_0)&\ge \ol{e}_\UEG(\rho_0,\nu_0) -\frac{C\delta}{\ell}(\rho_0^{5/3} + \rho_0^{4/3} + \rho_0\nu_0) - C\frac{\rho_0^{23/15} + \rho_0^{18/15}}{\ell^{2/5}},
\end{align*} 
\item[(iv)] The thermodynamic limit
$$
\lim_{\substack{\delta/\ell\to 0\\ \ell\delta\to\infty}} e_{\ell\tetra,\eta_\delta}(\rho_0,\vnu_0) = e_\UEG(\rho_0,\nu_0)
$$
exists and is independent of the choice of the regularization function $\eta$. 
\end{itemize}
\end{theorem}


\begin{proof}
\noindent\textbf{Proof of (i).} The upper bound in (i) is proved by
applying the lower bound of \cref{lboundthm} with the choice
\begin{equation}\label{rhojpchoice}
(\rho,\vjp):=(\iden_{\ell'\tetra}*\eta_{\delta'})(\rho_0,\tfrac{1}{2}\rho_0\vnu_0\times\cdot),
\end{equation}
where $\ell'\gg\ell$. Recalling~\eqref{etetradef}, we get
\begin{equation}\label{etetralow}
\begin{aligned}
&e_{\ell'\tetra,\eta_{\delta'}}(\rho_0,\vnu_0)\ge \frac{1}{|\ell'\tetra|} \sum_{\vz\in\ZZ^3} \sum_{j=1}^{24} \int_{\SO(3)}\dd\vR \dashint_{C_{\ell}}\dd\vtau  \times\\
\times \Bigg[ &(1-C\delta/\ell) E\bigl((\iden_{\ell\tetra_j}*\wt{\eta}_\delta)(\vR\cdot - \ell\vz - \vtau)(\iden_{\ell'\tetra}*\eta_{\delta'})(\rho_0,\tfrac{1}{2}\rho_0\vnu_0\times\cdot) \bigr)\\
&+\frac{C\delta}{\ell} T\bigl((\iden_{\ell\tetra_j}*\wt{\eta}_\delta)(\vR\cdot - \ell\vz - \vtau)(\iden_{\ell'\tetra}*\eta_{\delta'})(\rho_0,\tfrac{1}{2}\rho_0\vnu_0\times\cdot) \bigr)\Bigg] \\
&-\frac{1}{|\ell'\tetra|}\frac{\rho_0}{4} \int_{\RR^3} (\iden_{\ell'\tetra} * \eta_{\delta'})(\vx)  |\vnu_0\times\vx|^2\,\dd\vx  \\
& -\frac{C}{\ell|\ell'\tetra|} \int_{\RR^3} \Big( (1+\delta^{-1})(\iden_{\ell'\tetra}*\eta_{\delta'}) \rho_0 + \delta^3(\iden_{\ell'\tetra}*\eta_{\delta'})^2\rho_0^2 \Big)\\
&=:\mathrm{(I)} + \mathrm{(II)} + \mathrm{(III)}
\end{aligned}
\end{equation}
Here, (III) can be bounded simply by $-\frac{C\rho_0}{\ell}(1+\delta^{-1}+\delta^3\rho_0)$, using the trivial estimate:
\begin{equation}
\int_{\RR^3} (\iden_{\ell'\tetra}*\eta_{\delta'})^\alpha\le \int_{\RR^3} (\iden_{\ell'\tetra}*\eta_{\delta'}) = |\ell'\tetra|,
\end{equation}
valid for any $\alpha\ge 1$ due to the fact that $\iden_{\ell'\tetra}*\eta_{\delta'}\le 1$. 

Considering (I), we can proceed as follows. Define the index set
$$
J:=\left\{ \begin{aligned}(\vz,j)\in \ZZ^3\times\{1,\ldots,24\} : &\supp( (\iden_{\ell\tetra_j}*\wt{\eta}_\delta)(\vR\cdot - \ell\vz - \vtau)) \cap \supp(\iden_{\ell'\tetra}*\eta_{\delta'})\neq \emptyset,\\
&\text{for some}\;\vR\in\SO(3)\;\text{and}\; \vtau\in C_\ell \end{aligned} \right\},
$$
which collects all the small tetrahedra $\vR^\top(\ell\tetra_j + \ell\vz + \vtau)$ that possibly intersect with the big smeared tetrahedron $\supp(\iden_{\ell'\tetra}*\eta_{\delta'})$. Furthermore, define
$$
J_0:=\left\{ \begin{aligned}(\vz,j)\in \ZZ^3\times\{1,\ldots,24\} : &\supp( (\iden_{\ell\tetra_j}*\wt{\eta}_\delta)(\vR\cdot - \ell\vz - \vtau)) \subset (\ell'-\delta')\tetra,\\
&\text{for some}\;\vR\in\SO(3)\;\text{and}\; \vtau\in C_\ell \end{aligned} \right\},
$$
which contains small smeared tetrahedra which are well inside $\ell'\tetra$. See \cref{tilingdecomp}. 

Then, we may decompose the double sum in (I) into a sum over $J_0$ and one over $J\setminus J_0$. The latter part is treated around \cref{bdtetraboundst}, and the former part reads
\begin{multline}\label{etetraz}
\frac{1}{|\ell'\tetra|}  \sum_{(\vz,j)\in J_0} \int_{\SO(3)}\dd\vR \dashint_{C_{\ell}}\dd\vtau \Bigg[ (1-C\delta/\ell) E\bigl((\iden_{\ell\tetra_j}*\wt{\eta}_\delta)(\vR\cdot - \ell\vz - \vtau)(\rho_0,\tfrac{1}{2}\rho_0\vnu_0\times\cdot) \bigr)\\
+\frac{C\delta}{\ell} T\bigl((\iden_{\ell\tetra_j}*\wt{\eta}_\delta)(\vR\cdot - \ell\vz - \vtau)(\rho_0,\tfrac{1}{2}\rho_0\vnu_0\times\cdot) \bigr)\Bigg] .
\end{multline}
Note first that $\iden_{\ell\tetra_j}*\eta_\delta=(\iden_{\ell\tetra}*\eta_\delta)\circ \vT_{j,\ell}^{-1}$, where $\vT_{j,\ell}(\vx)=\vR_j\vx+\ell\vz_j$ is the isometry that puts $\vT_{j,\ell}(\ell\tetra)=\ell\tetra_j$. 
Setting 
$$
\vT(\vx)=\vT_{j,\ell}^{-1}\circ(\vR\vx - \ell\vz - \vtau)=\vR_j^{\top}\vR\vx - \vR_j^{\top}(\ell\vz+\vtau+\ell\vz_j).
$$
Using \cref{etranslemma} (i), we have
\begin{multline}
E\big(((\iden_{\ell\tetra}*\wt{\eta}_\delta) \circ \vT) (\rho_0,\tfrac{1}{2}\rho_0\vnu_0\times\cdot)\big)=E\big( (\iden_{\ell\tetra}*\wt{\eta}_\delta)(\rho_0,\tfrac{1}{2}\rho_0\vR_j^{\top}\vR\vnu_0\times\cdot)\big)\\
 + |\ell\tetra|\frac{\rho_0}{4}|\vR\vnu_0\times (\ell\vz+\vtau+\ell\vz_j)|^2 .
\end{multline}
Further, the same relation holds with $E$ replaced by $T$.

In other words, we expressed the energy in any (smeared) tetrahedron well inside $\ell'\tetra$ in terms of the energy in the (smeared) reference tetrahedron $\ell\tetra$, plus a \emph{gauge term}, which is position-dependent and grows faster than the volume. More concretely, \cref{etetraz} can be written as
\begin{multline}\label{welletetra}
\frac{M_\ell}{|\ell'\tetra||\ell\tetra|}\int_{\SO(3)}\dd\vR \Bigg[ (1-C\delta/\ell) E\bigl((\iden_{\ell\tetra}*\wt{\eta}_\delta)(\rho_0,\tfrac{1}{2}\rho_0\vR\vnu_0\times\cdot) \bigr)+\frac{C\delta}{\ell} T\bigl((\iden_{\ell\tetra}*\wt{\eta}_\delta)(\rho_0,\tfrac{1}{2}\rho_0\vR\vnu_0\times\cdot) \bigr)\Bigg]\\
+ \frac{1}{|\ell'\tetra|}\int_{\SO(3)}\dd\vR \dashint_{C_{\ell}}\dd\vtau \sum_{(\vz,j)\in J_0} |\ell\tetra|\frac{\rho_0}{4}|\vR\vnu_0\times (\ell\vz+\vtau+\ell\vz_j)|^2,
\end{multline}
where 
$$
M_\ell=|\tetra| |J_0|.
$$
The term involving $\frac{C\delta}{\ell}T(\ldots)$ can be bounded using \cref{kinweizgauge} (discarding all but the gauge term) as
\begin{equation}\label{etetralt}
\frac{M_\ell C\delta/\ell}{|\ell'\tetra||\ell\tetra|} \frac{\rho_0}{4}\int_{\SO(3)}\int_{\RR^3} (\iden_{\ell\tetra} * \wt{\eta}_{\delta})(\vx)  |\vR\vnu_0\times\vx|^2\,\dd\vx\dd\vR
\end{equation}

The gauge correction term (II) in \cref{etetralow} can be rewritten using the partition of unity 
$$
\int_{\SO(3)}\dd\vR\dashint_{C_{\ell}}\dd\vtau \sum_{\vz\in\ZZ^3} \sum_{j=1}^{24}  (\iden_{\ell\tetra_j}*\wt{\eta}_\delta)(\vR\cdot - \ell\vz-\vtau) \equiv 1,
$$
see \cref{pouind}, as
\begin{multline}\label{gaugepou}
\mathrm{(II)}=-\frac{1}{|\ell'\tetra|}\frac{\rho_0}{4}\int_{\SO(3)}\dd\vR\dashint_{C_{\ell}}\dd\vtau \sum_{(\vz,j)\in J} \int_{\RR^3} (\iden_{\ell\tetra_j}*\wt{\eta}_\delta)(\vR\vx - \ell\vz-\vtau) (\iden_{\ell'\tetra} * \eta_{\delta'})(\vx) |\vnu_0\times\vx|^2\,\dd\vx.
\end{multline}
Again, this sum may be split into two parts: the $J_0$ part can be written as
\begin{align*}
&-\frac{1}{|\ell'\tetra|}\frac{\rho_0}{4}\int_{\SO(3)}\dd\vR\dashint_{C_{\ell}}\dd\vtau \sum_{(\vz,j)\in J_0}\int_{\RR^3} (\iden_{\ell\tetra_j}*\wt{\eta}_\delta)(\vR\vx - \ell\vz-\vtau) |\vnu_0\times\vx|^2\,\dd\vx   \\
&= -\frac{1}{|\ell'\tetra|}\frac{\rho_0}{4}\int_{\SO(3)}\dd\vR\dashint_{C_{\ell}}\dd\vtau \sum_{(\vz,j)\in J_0} \int_{\RR^3} (\iden_{\ell\tetra}*\wt{\eta}_\delta)(\vx) |\vnu_0\times\vR^{\top}(\vR_j\vx + \ell\vz + \ell\vz_j + \vtau)|^2\,\dd\vx \\
&=-\frac{M_\ell}{|\ell'\tetra||\ell\tetra|}\frac{\rho_0}{4} \int_{\SO(3)} \int_{\RR^3} (\iden_{\ell\tetra}*\wt{\eta}_\delta)(\vx) |\vR\vnu_0\times \vx|^2 \,\dd\vx \dd\vR\\
&- \frac{1}{|\ell'\tetra|}\int_{\SO(3)}\dd\vR\dashint_{C_{\ell}}\dd\vtau \sum_{(\vz,j)\in J_0} |\ell\tetra|\frac{\rho_0}{4}  |\vR\vnu_0\times (\ell\vz + \ell\vz_j + \vtau)|^2,
\end{align*}
where in the last step we used once more that the smeared tetrahedron has barycenter 0, due to $\int_{\RR^3} \vx\eta(\vx)\,\dd\vx=0$.
Note that the last term cancels exactly with the second term of \cref{welletetra}. In total, the $J_0$ part of (I)+(II)  may be bounded by (using also \cref{etetralt})
$$
\frac{(1-C\delta/\ell) M_\ell}{|\ell'\tetra|}\int_{\SO(3)}  e_{\ell\tetra,\wt{\eta}_\delta}(\rho_0,\vR\vnu_0) \,\dd\vR.
$$

Next, we need to bound the part $J\setminus J_0$ of (I) corresponding to tetrahedra at the boundary. We use the Lieb--Oxford inequality and \cref{kinweizgauge}
(again dropping all but the gauge term) to get
\begin{equation}\label{bdtetraboundst}
\begin{aligned}
&\frac{1}{|\ell'\tetra|} \int_{\SO(3)}\dd\vR \dashint_{C_{\ell}}\dd\vtau \sum_{(\vz,j)\in J\setminus J_0} \times\\
	&\times \Bigg[ -(1-C\delta/\ell) c_\LO \rho_0^{4/3} \int_{\RR^3} (\iden_{\ell\tetra_j}*\wt{\eta}_\delta)^{4/3} (\vR\cdot - \ell\vz - \vtau) (\iden_{\ell'\tetra}*\eta_{\delta'})^{4/3} \\
	&+\frac{\rho_0}{4} \int_{\RR^3} (\iden_{\ell\tetra_j}*\wt{\eta}_\delta)(\vR\vx - \ell\vz - \vtau)(\iden_{\ell'\tetra}*\eta_{\delta'})(\vx) |\vnu_0\times\vx|^2\,\dd\vx \Bigg].
\end{aligned}
\end{equation}
Again, we see that the second term cancels exactly with the remaining, $J\setminus J_0$ part of the gauge correction term (II), see \cref{gaugepou}.
 The first term may be bounded by
$$
-\frac{(1-C\delta/\ell)M_\ell'}{|\ell'\tetra|} c_\LO \rho_0^{4/3}, \quad M_\ell'=|\tetra| |J\setminus J_0|.
$$
Let us bound the quantities $M_\ell$ and $M_\ell'$. 
If $\int_{\SO(3)} e_{\ell\tetra,\wt{\eta}_\delta}(\rho_0,\vR\vnu_0)\,\dd\vR\le 0$, then we can use the trivial bound $M_\ell/|\ell'\tetra|\le 1$. 
The boundary tetrahedra indexed by $J\setminus J_0$ are at a distance $\OC(\ell+\delta+\delta')$ from $\partial(\ell'\tetra)$, and hence they fill a volume of $M_\ell'=\OC(\ell'^2(\ell+\delta+\delta'))$. But then $M_\ell\ge |\ell'\tetra| - C\ell'^2(\ell+\delta+\delta')$, so we arrive at the bound
\begin{multline}\label{etetralowinter}
e_{\ell'\tetra,\eta_{\delta'}}(\rho_0,\vnu_0)\ge \left(1-C\sigma\frac{\delta}{\ell}-C\sigma\frac{\ell+\delta+\delta'}{\ell'}\right) \int_{\SO(3)} e_{\ell\tetra,\wt{\eta}_\delta}(\rho_0,\vR\vnu_0)\,\dd\vR\\
 -\frac{\ell+\delta+\delta'}{\ell'} \rho_0^{4/3}-\frac{C\rho_0}{\ell}(1+\delta^{-1}+\delta^3\rho_0),
\end{multline}
where we have set 
\begin{equation}\label{sigmadef}
\sigma=\begin{cases}
0 & \text{if} \int_{\SO(3)} e_{\ell\tetra,\wt{\eta}_\delta}(\rho_0,\vR\vnu_0)\,\dd\vR\le 0\\
1 & \text{otherwise}
\end{cases}
\end{equation}

 Next, we use \cref{ethm} to bound part of the l.h.s. of the preceding inequality. To this end,
note that $\vv(\vx)=\iden_{\supp(\iden_{\ell'\tetra} * \eta_{\delta'})}(\vx) \frac{1}{2} \vnu_0\times\vx$ is clearly divergence-free and has 
$\int_{\RR^3} \rho|\vD\vv|=|\ell'\tetra| \frac{\sqrt{2}}{2}\rho_0|\vnu_0|$. We have
\begin{align*}
&e_{\ell'\tetra,\eta_{\delta'}}(\rho_0,\vnu_0)= (1-C\sigma \delta/\ell) e_{\ell'\tetra,\eta_{\delta'}}(\rho_0,\vnu_0) + \frac{C\sigma \delta}{\ell} e_{\ell'\tetra,\eta_{\delta'}}(\rho_0,\vnu_0) \\
&= (1-C\sigma \delta/\ell)  e_{\ell'\tetra,\eta_{\delta'}}(\rho_0,\vnu_0)\\
& + \frac{1}{|\ell'\tetra|}\frac{C\sigma \delta}{\ell}\Bigg( C \rho_0^{5/3} \int_{\RR^3} (\iden_{\ell'\tetra} * \eta_{\delta'})^{5/3} + C\rho_0 \int_{\RR^3}|\grad\sqrt{\iden_{\ell'\tetra} * \eta_{\delta'}}|^2+ C|\ell'\tetra| \rho_0\nu_0\Bigg)\\
&\le (1-C\sigma \delta/\ell) e_{\ell'\tetra,\eta_{\delta'}}(\rho_0,\vnu_0) + \frac{C\sigma \delta}{\ell} \left( \rho_0^{5/3} + \frac{\rho_0}{\ell'\delta'} + \rho_0\nu_0 \right),
\end{align*}
where we also used the fact that $\frac{1}{\ell'^3} \int_{\RR^3} |\grad\sqrt{\iden_{\ell'\tetra} * \eta_{\delta'}}|^2 =\OC\left( \frac{1}{\ell'\delta'} \right)$.
We obtain
\begin{align*}
&\left(1-C\sigma\frac{\delta}{\ell}-C\sigma\frac{\ell+\delta+\delta'}{\ell'}\right) \int_{\SO(3)} e_{\ell\tetra,\wt{\eta}_\delta}(\rho_0,\vR\vnu_0)\,\dd\vR \le\left(1-C\sigma\frac{\delta}{\ell} \right) e_{\ell'\tetra,\eta_{\delta'}}(\rho_0,\vnu_0) \\
&+ \frac{C\rho_0}{\ell}\left( 1+\delta^{-1}+\delta^3\rho_0 + \delta \rho_0^{2/3} + \frac{\delta}{\ell'\delta'} + \delta \nu_0 \right) + \frac{\ell+\delta+\delta'}{\ell'} \rho_0^{4/3}.
\end{align*}
After taking the liminf as $\ell'\to\infty$ and dividing by $(1-C\sigma\delta/\ell)$, we arrive at
$$
\int_{\SO(3)} e_{\ell\tetra,\wt{\eta}_\delta}(\rho_0,\vR\vnu_0)\,\dd\vR \le \ul{e}_\UEG(\rho_0,\nu_0) + \frac{C\rho_0}{\ell}\Big( 1+\delta^{-1}+\delta^3\rho_0 + \delta \rho_0^{2/3} + \delta\nu_0\Big)
$$
as desired.\\


\noindent\textbf{Proof of (ii).} The proof of the lower bound is similar. We use the upper bound of \cref{eneup} with the same choice \cref{rhojpchoice}. We have for any $0<\wt{\delta}<\frac{\ell}{2}$ and $0<\alpha\le\frac{1}{2}$, 
\begin{equation}\label{etetraup}
\begin{aligned}
&e_{\ell'\tetra,\eta_{\delta'}}(\rho_0,\vnu_0)\le\frac{1}{|\ell'\tetra|}\dashint_{1-\alpha}^{1+\alpha}\frac{\dd t}{t^4} \int_{\SO(3)}\dd\vR \dashint_{C_{t\ell}}\dd\vtau
\sum_{\vz\in\ZZ^3} \sum_{j=1}^{24}\times\\
	&\times E\bigl(\beta^{-3} (\iden_{t\ell \vT_j \beta\tetra} * \wt{\eta}_{t\wt{\delta}})(\vR\cdot - t\ell\vz - \vtau)(\iden_{\ell'\tetra}*\eta_{\delta'})(\rho_0,\tfrac{1}{2}\rho_0\vnu_0\times\cdot)\bigr)\\
	&-\frac{1}{|\ell'\tetra|}\frac{\rho_0}{4} \int_{\RR^3} (\iden_{\ell'\tetra} * \eta_{\delta'})(\vx)  |\vnu_0\times\vx|^2\,\dd\vx  \\
	&  + C\delta^2\log(\alpha^{-1}) \rho_0^2
\end{aligned}
\end{equation}
where we used the notation $\beta=1-\wt{\delta}/\ell$. The smearing parameter $\tilde{\delta}$ is chosen so that $\wt{\delta}=\beta\delta$, i.e. 
$\wt{\delta}=\delta/(1+\delta/\ell)$.

Define $J$ and $J_0$ analogously. First, we consider the summation over $J_0$ in the first term.
Using $t\ell \vT_j(1-\wt{\delta}/\ell)\tetra=t\ell\beta\vR_j\tetra + t\ell\vz_j$, we may write
\begin{align*}
(\iden_{t\ell \vT_j\beta\tetra} * \wt{\eta}_{t\wt{\delta}})(\vR\vx - t\ell\vz - \vtau)&=(\iden_{t\ell\beta\tetra} * \wt{\eta}_{t\wt{\delta}})(\vT(\vx))=(\iden_{t\ell\tetra} * \wt{\eta}_{t\delta})(\beta^{-1}\vT(\vx)),
\end{align*}
where $\vT$ is defined as above except with $t\ell$ in place of $\ell$. Here, we used the scaling relation
\begin{equation}\label{cutoffscale}
(\iden_{\ell/a\tetra}*\wt{\eta}_{\delta/a})(\vx/a)=(\iden_{\ell\tetra}*\wt{\eta}_\delta)(\vx)
\end{equation}
valid for all $a>0$.
We have
\begin{equation}\label{etetraup0}
\begin{aligned}
&\frac{1}{|\ell'\tetra|}\dashint_{1-\alpha}^{1+\alpha}\frac{\dd t}{t^4} \int_{\SO(3)}\dd\vR \dashint_{C_{t\ell}}\dd\vtau\sum_{(\vz,j)\in J_0}
 E\bigl(\vx\mapsto \beta^{-3} (\iden_{t\ell\tetra} * \wt{\eta}_{t\delta})(\beta^{-1}\vT(\vx)) (\rho_0,\tfrac{1}{2}\rho_0\vnu_0\times\vx)\bigr)\\
=&\frac{1}{|\ell'\tetra|}\dashint_{1-\alpha}^{1+\alpha}\frac{\dd t}{t^4} \int_{\SO(3)}\dd\vR  \dashint_{C_{t\ell}}\dd\vtau\sum_{(\vz,j)\in J_0} \times\\
	&\times\Bigg[ E\bigl(\vx\mapsto \beta^{-3} (\iden_{t\ell\tetra} * \wt{\eta}_{t\delta})(\beta^{-1}\vx) (\rho_0,\tfrac{1}{2}\rho_0\vR_j^{-1}\vR\vnu_0\times\vx)\bigr)\\
	&+|t\ell\tetra| \frac{\rho_0}{4} |\vR\vnu_0\times (t\ell\vz + \vtau + t\ell\vz_j)|^2 \Bigg]\\
\le &\int_{\SO(3)}\dd\vR  \dashint_{1-\alpha}^{1+\alpha}\frac{\dd t}{t^4} \frac{M_{t\ell}}{|\ell'\tetra||t\ell\tetra|} \beta^{-1} E\bigl(\vx\mapsto (\iden_{t\ell\tetra} * \wt{\eta}_{t\delta})(\vx) (\rho_0,\tfrac{1}{2}\rho_0\beta\vR\vnu_0\times\vx)\bigr)\\
	&+\frac{1}{|\ell'\tetra|}\int_{\SO(3)}\dd\vR  \dashint_{1-\alpha}^{1+\alpha}\frac{\dd t}{t^4} \int_{C_{t\ell}}\dd\vtau\sum_{(\vz,j)\in J_0} |t\ell\tetra| \frac{\rho_0}{4} 
	|\vR\vnu_0\times (t\ell\vz + \vtau + t\ell\vz_j)|^2 
\end{aligned}
\end{equation}
where in the first step we used \cref{etranslemma}, and \cref{Escale} in the second (with $\lambda=\beta^{-3}$ and dropping the gauge terms).

Furthermore, note that the function $\vnu_0\mapsto E((\iden_{t\ell\tetra}*\wt{\eta}_{t\delta})(\rho_0,\tfrac{1}{2}\rho_0\vR\vnu_0\times\cdot))$ is convex for all fixed $\rho_0$
and $\vR$, using the convexity of $F$.
Hence, we find
\begin{align*}
\beta^{-1}E((\iden_{t\ell\tetra}*\wt{\eta}_{t\delta})(\rho_0,\tfrac{1}{2}\rho_0\beta\vR\vnu_0\times\cdot))\le 
 &E((\iden_{t\ell\tetra}*\wt{\eta}_{t\delta})(\rho_0,\tfrac{1}{2}\rho_0\vR\vnu_0\times\cdot)) \\
 &+ \frac{1-\beta}{\beta} E((\iden_{t\ell\tetra}*\wt{\eta}_{t\delta})\rho_0),
\end{align*}
where used $0<\beta<\frac{1}{2}$.
Here, we used that $E(\rho,0)=E(\rho)$, see \cref{corrthm}.

The gauge correction term in \cref{etetraup} may be written similarly to \cref{gaugepou}. The $J_0$ part reads
\begin{align*}
&-\frac{\rho_0}{4} \dashint_{1-\alpha}^{1+\alpha}\frac{\dd t}{t^4}\frac{M_{t\ell}}{|\ell'\tetra||t\ell\tetra|} \int_{\RR^3} (\iden_{t\ell\tetra}*\wt{\eta}_{t\delta})(\vx) |\vnu_0\times \vx|^2 \,\dd\vx\\
&- \frac{1}{|\ell'\tetra|}\dashint_{1-\alpha}^{1+\alpha}\frac{\dd t}{t^4} \int_{\SO(3)}\dd\vR\dashint_{C_{t\ell}}\dd\vtau \sum_{(\vz,j)\in J_0}  |t\ell\tetra|\frac{\rho_0}{4}  |\vR\vnu_0\times (t\ell\vz + t\ell\vz_j + \vtau)|^2.
\end{align*}
We see that the second term here cancels with the last term of \cref{etetraup0}.

Next, we consider the $J\setminus J_0$ part. We use \cref{ethm} to bound the first term of \cref{etetraup}, 
\begin{align*}
&\frac{1}{|\ell'\tetra|}\dashint_{1-\alpha}^{1+\alpha}\frac{\dd t}{t^4} \int_{\SO(3)}\dd\vR  \dashint_{C_{t\ell}}\dd\vtau
\sum_{(\vz,j)\in J\setminus J_0} \times\\
	&\times E\bigl(\beta^{-3} (\iden_{t\ell \vT_j \beta\tetra} * \wt{\eta}_{t\wt{\delta}})(\vR\cdot - t\ell\vz - \vtau)(\iden_{\ell'\tetra}*\eta_{\delta'})(\rho_0,\tfrac{1}{2}\rho_0\vnu_0\times\cdot)\bigr)\\
&\le\frac{1}{|\ell'\tetra|}\dashint_{1-\alpha}^{1+\alpha}\frac{\dd t}{t^4} \int_{\SO(3)}\dd\vR \dashint_{C_{t\ell}}\dd\vtau
\sum_{(\vz,j)\in J\setminus J_0} \times\\
	&\times \Bigg[ C\rho_0^{5/3} \beta^{-5}\int_{\RR^3} (\iden_{t\ell \vT_j \beta\tetra} * \wt{\eta}_{t\wt{\delta}})(\vR\cdot - t\ell\vz - \vtau)\\
	&+ C\rho_0 \beta^{-3} \int_{\RR^3} \left|\grad \sqrt{(\iden_{t\ell \vT_j \beta\tetra} * \wt{\eta}_{t\wt{\delta}})(\vR\cdot - t\ell\vz - \vtau)(\iden_{\ell'\tetra}*\eta_{\delta'})}\right|^2\\
	&+ \beta^{-3}\frac{\rho_0}{4} \int_{\RR^3} (\iden_{t\ell \vT_j \beta\tetra} * \wt{\eta}_{t\wt{\delta}})(\vR\vx - t\ell\vz - \vtau)(\iden_{\ell'\tetra}*\eta_{\delta'})(\vx) |\vnu_0\times\vx|^2 \,\dd\vx\\
	&+ C\beta^{-3}\rho_0\nu_0\int_{\RR^3}(\iden_{t\ell \vT_j \beta\tetra} * \wt{\eta}_{t\wt{\delta}})(\vR\cdot - t\ell\vz - \vtau) \Bigg]
\end{align*}
For the $J\setminus J_0$ part of the gauge correction term, we use the partition of unity with holes to cancel the gauge term in the previous estimate.
In total, we have obtained
\begin{equation}\label{etetraavgup}
\begin{aligned}
e_{\ell'\tetra,\eta_{\delta'}}(\rho_0,\vnu_0)&\le\left(1+C\sigma\frac{\ell+\delta+\delta'}{\ell'}\right)\left(\dashint_{1-\alpha}^{1+\alpha}\int_{\SO(3)}e_{t\ell\tetra,\wt{\eta}_{t\delta}}(\rho_0,\vR\vnu_0)\,\dd\vR \frac{\dd t}{t^4} + 
\frac{\delta}{\ell} \dashint_{1-\alpha}^{1+\alpha}e_{t\ell\tetra,\wt{\eta}_{t\delta}}(\rho_0,0)\,\frac{\dd t}{t^4}\right)\\
& +C\frac{\ell+\delta+\delta'}{\ell'}\rho_0\left(\rho_0^{2/3}+(\ell\delta)^{-1} + \nu_0\right) + \frac{C\rho_0}{\delta'\ell'} + C\delta^2\log(\alpha^{-1}) \rho_0^2\\
&\le \left(1+C\sigma\frac{\ell+\delta+\delta'}{\ell'}\right)\left( \dashint_{1-\alpha}^{1+\alpha} \int_{\SO(3)} e_{t\ell\tetra,\wt{\eta}_{t\delta}}(\rho_0,\vR\vnu_0)\,\dd\vR \frac{\dd t}{t^4}  + 
\frac{C\rho_0}{\ell}\big(\delta\rho_0^{2/3} + \ell^{-1}\big)\right)\\
& +C\frac{\ell+\delta+\delta'}{\ell'}\rho_0\left(\rho_0^{2/3}+(\ell\delta)^{-1} + \nu_0\right) + \frac{C\rho_0}{\delta'\ell'} + C\delta^2\log(\alpha^{-1}) \rho_0^2.
\end{aligned}
\end{equation}
The proof of (ii) is finished after taking the limsup as $\ell'\delta'\to\infty$, $\delta'/\ell'\to 0$.\\


\noindent\textbf{Proof of (iii).} Finally, to show (iii), we substitute $\ell\to t\ell$ and $\delta\to t\delta$ in \cref{etetralowinter}, apply $\dashint_{1-\alpha}^{1+\alpha}\frac{\dd t}{t^4}$ to the inequality
and use (ii), with the roles of $\eta$ and $\wt{\eta}$ interchanged,
\begin{align*}
e_{\ell'\tetra,\wt{\eta}_{\delta'}}(\rho_0,\vnu_0)&\ge \left(1-C\sigma\frac{\delta}{\ell}-C\sigma\frac{\ell+\delta+\delta'}{\ell'}\right) \dashint_{1-\alpha}^{1+\alpha} \int_{\SO(3)} e_{t\ell\tetra,\eta_{t\delta}}(\rho_0,\vR\vnu_0)\,\dd\vR\frac{\dd t}{t^4} \\
& -C\frac{\ell+\delta+\delta'}{\ell'} \rho_0^{4/3}-\frac{C\rho_0}{\ell}(1+\delta^{-1}+\delta^3\rho_0)\\
&\ge \ol{e}_\UEG(\rho_0,\nu_0) -C\frac{\ell+\delta+\delta'}{\ell'} \big(\rho_0^{5/3}+ \rho_0^{4/3} + \rho_0\nu_0\big)\\
&-\frac{C\rho_0}{\ell}\big(1+\delta^{-1}+\delta^3\rho_0 + \delta\rho_0^{2/3}  + \delta\rho_0\nu_0 \big) - C\delta^2\rho_0^2.
\end{align*}
Here, we have modified $\sigma$ in suitable way (cf. \cref{sigmadef}).
The choice $\delta=\ell^{-1/3}\rho_0^{-4/9}$ leads to
\begin{align*}
e_{\ell'\tetra,\wt{\eta}_{\delta'}}(\rho_0,\vnu_0)&\ge \ol{e}_\UEG(\rho_0,\nu_0) -C\frac{\ell+\ell^{-1/3}\rho_0^{-4/9}+\delta'}{\ell'} \big(\rho_0^{5/3}+ \rho_0^{4/3} + \rho_0\nu_0\big)\\
&-C\frac{\rho_0^{13/9} + \rho_0^{10/9}}{\ell^{2/3}} - C\frac{\rho_0^{11/9} + \nu_0\rho_0^{14/9}}{\ell^{4/3}} - \frac{C\rho_0^{2/3}}{\ell^2} - \frac{C\rho_0}{\ell}
\end{align*}
and setting $\ell=(\ell')^{3/5} \rho_0^{-2/15}$ results in
\begin{align*}
e_{\ell'\tetra,\wt{\eta}_{\delta'}}(\rho_0,\vnu_0)&\ge \ol{e}_\UEG(\rho_0,\nu_0) -\frac{C\delta'}{\ell'}(\rho_0^{5/3} + \rho_0^{4/3} + \rho_0\nu_0) - C\frac{\rho_0^{23/15} + \rho_0^{18/15}}{(\ell')^{2/5}},
\end{align*}
using $\ell'\rho_0^{1/3}\ge C$.

\noindent\textbf{Proof of (iv).} To show that the limit exists, it is enough to prove that $\ol{e}_\UEG(\rho_0,\nu_0)\le\ul{e}_\UEG(\rho_0,\nu_0)$.
But this follows from (i) combined with (iii). It remains to prove that $e_\UEG(\rho_0,\nu_0)$ is independent of the choice of the regularization function.
Let $\wt{e}_\UEG(\rho_0,\nu_0)$ denote $e_\UEG(\rho_0,\nu_0)$, but with $\eta_\delta$ replaced 
with $\wt{\eta}_\delta$. Using (iii) and then (i), we find that the following chain of inequalities hold true,
\begin{align*}
e_\UEG(\rho_0,\nu_0)\le \int_{\SO(3)} e_{\ell\tetra,\wt{\eta}_{\delta}}(\rho_0,\vR\vnu_0)\,\dd\vR + \mathrm{(error)} \le e_\UEG(\rho_0,\nu_0)  + \mathrm{(error)},
\end{align*}
from which we deduce $\wt{e}_\UEG(\rho_0,\nu_0)=e_\UEG(\rho_0,\nu_0)$.
\end{proof}

By omitting the Coulomb interaction, we find the following.

\begin{proposition}
Fix $\rho_0\in\RR_+$ and $\nu_0\in\RR_+$. Let $\eta$ and $\wt{\eta}$ be regularization functions.  Let $\ul{\tau}_\UEG(\rho_0,\nu_0)$ and $\ol{\tau}_\UEG(\rho_0,\nu_0)$ denote the liminf and the limsup of $\tau_{\ell\tetra,\eta_\delta}(\rho_0,\nu_0)$
as $\ell\delta\to\infty$ and $\delta/\ell\to 0$. 
\begin{itemize}
\item[(i)] For $C\delta\le\ell$, we have
$$
\tau_{\ell\tetra,\wt{\eta}_\delta}(\rho_0,\nu_0)\le \ul{\tau}_\UEG(\rho_0,\nu_0) + \frac{C\rho_0}{\ell}\Big( \delta^{-1} + \delta \rho_0^{2/3} + \delta\nu_0\Big)
$$
\item[(ii)] We have 
$$
\tau_{\ell\tetra,\wt{\eta}_\delta}(\rho_0,\nu_0)\ge \ol{\tau}_\UEG(\rho_0,\nu_0) -
\frac{C\rho_0}{\ell}\big(\delta\rho_0^{2/3} + \ell^{-1}\big) 
$$
\item[(iii)] The thermodynamic limit
$$
\lim_{\substack{\delta/\ell\to 0\\ \ell\delta\to\infty}} \tau_{\ell\tetra,\eta_\delta}(\rho_0,\vnu_0) = \tau_\UEG(\rho_0,\nu_0)
$$
exists and is independent of the choice of the regularization function $\eta$. 
\end{itemize}
\end{proposition}

\subsection{Proof of \cref{thermothm}}

The proof is very similar to the one in \cite{lewin2020local}. Let $\{\Omega_N\}$ be a sequence of domains such that $|\Omega_N|\to\infty$, $|\partial\Omega_N+B_r|\le Cr|\Omega_N|^{2/3}$ for $r\le|\Omega_N|^{1/3}/C$
and $\delta_N/|\Omega_N|^{1/3}\to 0$, $\delta_N|\Omega_N|^{1/3}\to\infty$.

We can define the index sets $J$ and $J_0$ for the domain $\Omega_N$ in an obvious manner,
$$
J:=\left\{ \begin{aligned}
(\vz,j)\in \ZZ^3\times\{1,\ldots,24\} : &\supp( (\iden_{\ell\tetra_j}*\eta_\delta)(\vR\cdot - \ell\vz - \vtau)) \cap \supp(\iden_{\Omega_N}*\eta_{\delta_N})\neq \emptyset,\\
&\text{for some}\,\vR\in\SO(3)\,\text{and}\,\vtau\in C_\ell
\end{aligned}
 \right\},
$$
which collects all the small tetrahedra $\vR^\top(\ell\tetra_j + \ell\vz + \vtau)$ that possibly intersect the smearing of $\Omega_N$. Furthermore, define
$$
J_0:=
\left\{ \begin{aligned}
(\vz,j)\in \ZZ^3\times\{1,\ldots,24\} : &\supp( (\iden_{\ell\tetra_j}*\eta_\delta)(\vR\cdot - \ell\vz - \vtau)) \subset \Omega_N\setminus(\partial\Omega_N+B_{\delta_N}),\\
&\text{for some}\,\vR\in\SO(3)\,\text{and}\,\vtau\in C_\ell
\end{aligned}
 \right\},
$$
which contains small smeared tetrahedra which are well inside $\Omega_N$.

The argument for the tetrahedron $\ell'\tetra$ works exactly the same for the general domain $\Omega_N$,
we only need to bound the quantities $M_\ell$ and $M_\ell'$ to obtain the analog of \cref{etetralowinter}.
If $e_{\Omega_N,\delta_N}(\rho_0,\vnu_0)\le 0$, then we can use the trivial bound $M_\ell/|\Omega_N|\le 1$. 
The boundary tetrahedra indexed by $J\setminus J_0$ are at a distance $r=\ell+\delta+\delta_N$ and hence, by the uniform regularity hypothesis, they fill a volume of 
$$
M_\ell'=C|\partial\Omega_N+B_r|\le C(\ell+\delta+\delta_N)|\Omega_N|^{2/3}
$$
for sufficiently large $N$, since $r\le |\Omega_N|^{1/3}/C$ by the convergence $\delta_N/|\Omega_N|^{1/3}\to 0$.
But then $M_\ell\ge |\Omega_N| - C(\ell+\delta+\delta')|\Omega_N|^{2/3}$, so we arrive at the bound
\begin{align*}
e_{\Omega_N,\delta_N}(\rho_0,\vnu_0)&\ge \left(1-C\sigma\frac{\delta}{\ell}-C\sigma\frac{\ell+\delta+\delta_N}{|\Omega_N|^{1/3}}\right)\int_{\SO(3)} e_{\ell\tetra,\delta}(\rho_0,\vR\vnu_0)\,\dd\vR\\
& -\frac{\ell+\delta+\delta_N}{|\Omega_N|^{1/3}} \rho_0^{4/3}-\frac{C\rho_0}{\ell}(1+\delta^{-1}+\delta^3\rho_0).
\end{align*}
Taking $\delta$ fixed and $\ell=|\Omega_N|^{1/6}$, the r.h.s. goes to $e_\UEG(\rho_0,\nu_0)$ as $\delta_N/|\Omega_N|^{1/3}\to 0$. 

For the upper bound, we use the analog of \cref{etetraavgup}, with $\alpha=1/2$,
\begin{align*}
e_{\Omega_N,\delta_N}(\rho_0,\vnu_0)&\le \left(1+C\sigma\frac{\ell+\delta+\delta_N}{|\Omega_N|^{1/3}}\right)\left( \dashint_{1/2}^{3/2}\int_{\SO(3)} e_{t\ell\tetra,t\delta}(\rho_0,\vR\vnu_0)\,\dd\vR\frac{\dd t}{t^4}  + 
\frac{C\rho_0}{\ell}\big(\delta\rho_0^{2/3} + \ell^{-1}\big)\right)\\
& +C\frac{\ell+\delta+\delta_N}{|\Omega_N|^{1/3}}\rho_0\left(\rho_0^{2/3}+(\ell\delta)^{-1} + \nu_0\right) + \frac{C\rho_0}{\delta_N|\Omega_N|^{1/3}} + C\delta^2 \rho_0^2.
\end{align*}
Using now $\delta_N|\Omega_N|^{1/3}\to\infty$ as well as $\delta_N/|\Omega_N|^{1/3}\to 0$, and the choice $\ell=|\Omega_N|^{1/6}$, $\delta=|\Omega_N|^{-1/12}$, we see that the r.h.s. goes to  $e_\UEG(\rho_0,\nu_0)$, and the proof is finished.

The proof for the kinetic energy per volume $\tau_{\Omega_N,\delta_N}(\rho_0,\nu_0)$ is very similar. The limit for $e^\xc_{\Omega_N,\delta_N}(\rho_0,\nu_0)$ is obtained by subtracting 
$\tau_{\Omega_N,\delta_N}(\rho_0,\nu_0)$ from $e_{\Omega_N,\delta_N}(\rho_0,\nu_0)$ and taking the limit term-by-term.

\subsection{Further proofs}

\begin{proof}[Proof of \cref{basicprop}]
\noindent\textbf{Part (i).} Following \cite{sen2018local}, it is convenient to consider the so-called \emph{intrinsic kinetic energy tensor},
\begin{equation}\label{intrkin}
\vomega_\gamma(\vx):=(\vpi_\vx\otimes\ol{\vpi_\vy}) \gamma(\vx,\vy)|_{\vy=\vx}=\vtau_\gamma(\vx) - \frac{\vzeta_\gamma(\vx)  \otimes  \ol{\vzeta_\gamma(\vx)}}{\rho_\gamma(\vx)},
\end{equation}
where $\vpi_\vx:=-i\grad_\vx - \frac{\vzeta_\gamma(\vx)}{\rho_\gamma(\vx)}$. We immediately have that $0\le \vomega_\gamma(\vx)=\vomega_\gamma(\vx)^\dag$.
 Also, $\vomega_\gamma$ is invariant under gauge transformations.
Note first that
\begin{equation}\label{omegaexpand}
\begin{aligned}
&\frac{\vomega(\vx)-\vomega(\vx)^\top}{2i}=\frac{1}{2i}(\vpi_\vx\otimes\ol{\vpi}_\vy-\ol{\vpi}_\vy\otimes\vpi_\vx)\gamma(\vx,\vy)|_{\vy=\vx}\\
&=\frac{1}{2i}\sum_{j\ge 1} \lambda_j  \bigl(\grad\phi_j(\vx)\otimes \ol{\grad\phi_j(\vx)} - \ol{\grad\phi_j(\vx)} \otimes \grad\phi_j(\vx)\bigr)+ \frac{1}{2i}\frac{\vzeta(\vx)\otimes\ol{\vzeta(\vx)}-\ol{\vzeta(\vx)}\otimes\vzeta(\vx)}{\rho(\vx)}\\
&=\vD_a\vjp(\vx) - \frac{1}{2}\frac{\grad\rho(\vx) \otimes \vjp(\vx) - \vjp(\vx)\otimes\grad\rho(\vx)}{\rho(\vx)}\\
&=\rho(\vx)\vD_a\frac{\vjp(\vx)}{\rho(\vx)}.
\end{aligned}
\end{equation}
We find
$$
\rho(\vx)\left|\vD_a\frac{\vjp(\vx)}{\rho(\vx)}\right|\le \sqrt{d} \Tr_{\RR^d} \vomega(\vx),
$$
which upon expansion of the r.h.s.,
$$
\Tr_{\RR^d} \vomega(\vx)=\tau(\vx) - \frac{|\vzeta(\vx)|^2}{\rho(\vx)}=\tau(\vx) - |\grad\sqrt{\rho}(\vx)|^2 - \frac{|\vjp(\vx)|^2}{\rho(\vx)}
$$
yields the desired bound.

\noindent\textbf{Part (ii).} Follows from the Cauchy--Bunyakovsky--Schwarz inequality.
\end{proof}

\begin{proof}[Proof of \cref{kinweizgauge}]
Let $\gamma(\vx,\vy)=\sum_{j\ge 1} \lambda_j \phi_j(\vx)\ol{\phi_j(\vy)}$,
for some $0\le \lambda_j\le 1$ with $\sum_{j\ge 1}\lambda_j=N$ and $\{\phi_j\}_{j\ge 1}\subset H^1(\RR^d)$ $L^2$-orthonormal.
Let $\{\phi_j^n\}\subset C_c^\infty(\RR^3)$ be such that $\|\phi_j^n\|=1$ and $\phi_j^n\to\phi_j$ in $H^1$ as $n\to\infty$. 
Denoting $\rho_n(\vx)=\sum_{j\ge 1} \lambda_j|\phi_j^n(\vx)|^2$, we have that $\|\sqrt{\rho}-\sqrt{\rho_n}\|\to 0$
and $\|\grad \sqrt{\rho}-\grad \sqrt{\rho_n}\|\to 0$ by an argument similar to proof of \cite[Theorem 1.3]{lieb1983density}.
It is easy to see that we also have $\|\frac{\vjp_n}{\sqrt{\rho_n}}-\frac{\vjp}{\sqrt{\rho}}\|\to 0$. Using \cref{basicprop} (i), we may write
$$
\sum_{j\ge 1} \lambda_j \|\grad\phi_j^n\|^2\ge \int_{\RR^3} |\grad\sqrt{\rho_n}|^2 + \int_{\RR^3} \frac{|\vjp_n|^2}{\rho_n} + \frac{1}{\sqrt{6}} \int_{\RR^3} \rho_n|\vnu_n|,
$$
where $\vjp_n$ and $\vnu_n$ are defined in terms of the $\phi_j^n$. Using expansion \cref{omegaexpand}, 
\begin{align*}
&\int_{\RR^3} \left| \rho_n\vnu_n - \rho\vnu\right|=\int_{\RR^3} \left| \rho_n\rot \frac{\vjp_n}{\rho_n} - \rho\rot\frac{\vjp}{\rho}\right|\\
&\le \sum_{j\ge 1} \lambda_j \int_{\RR^3} |\grad\phi_j^n\times\ol{\grad\phi_j^n}-\grad\phi_j\times\ol{\grad\phi_j}|+\int_{\RR^3} \left| \frac{\grad\rho_n}{\sqrt{\rho_n}} \times \frac{\vjp_n}{\sqrt{\rho_n}} - \frac{\grad\rho}{\sqrt{\rho}} \times \frac{\vjp}{\sqrt{\rho}}\right|\\
&\le \sum_{j\ge 1} \lambda_j  (\|\grad\phi_j^n\|+\|\grad\phi_j^n\|) \|\grad\phi_j^n-\grad\phi_j\| \\
&+ 2\sum_{j\ge 1} \lambda_j \left[  \left(\int_{\RR^3} \frac{|\vjp_n|^2}{\rho_n}\right)^{1/2} \|\grad\sqrt{\rho_n}-\grad\sqrt{\rho}\|
+ \|\grad\sqrt{\rho}\| \left\|\frac{\vjp_n}{\sqrt{\rho_n}}-\frac{\vjp}{\sqrt{\rho}}\right\|\right]\to 0,
\end{align*}
we obtain the desired bound. 
\end{proof}

\bibliographystyle{siamplain}
\bibliography{dft}

\end{document}